\pdfoutput=1
\documentclass[11pt]{article}
\usepackage{amsthm, amsmath, amssymb}
\usepackage{graphicx}
\usepackage[tight]{subfigure}

\usepackage{algorithm}
\usepackage{algorithmic}

\usepackage{tikz}
\usetikzlibrary{matrix}

\usepackage[sort&compress,numbers,merge]{natbib}

\newtheorem{definition}{Definition}[section]
\newtheorem{proposition}{Proposition}[section]

\newcommand{\la}{\lambda}
\newcommand{\s}{\sigma}

\newcommand{\R}{\mathbb{R}}

\newcommand{\W}{\mathcal{W}}

\def\l{\langle}

\def\r{\rangle}

\newcommand{\weight}{w}
\newcommand{\estweight}{\widehat{w}}

\newcommand{\fspace}{\mathcal{H}}

\newcommand{\fmap}{\psi}
\newcommand{\fdim}{M}

\newcommand{\sysop}{\mathcal{A}}
\newcommand{\estsysop}{\widehat{A}}
\newcommand{\controlop}{\mathcal{B}}
\newcommand{\estcontrolop}{\widehat{B}}
\newcommand{\measop}{\mathcal{K}}

\newcommand{\data}{X}
\newcommand{\centers}{C}
\newcommand{\centerscontrol}{D}

\newcommand{\control}{\delta}

\newcommand{\kernel}{k}

\newcommand{\eval}{\la}

\newcommand{\empK}{\ensuremath{K}}

\newcommand{\dom}{\Omega}
\newcommand{\sampSet}{\mathcal{X}}

\newcommand{\sampSetLong}{\{x_1, \dots, x_{\nsamp}\}}

\newcommand{\nsamp}{N}
\newcommand{\ncent}{M}

\newcommand{\shCent}{\mathcal{C}}
\newcommand{\shCentLong}{\{c_1,\dots,c_{\ncent}\}}

\newcommand{\eqlabel}[1]{\label{eq:#1}}

\renewcommand{\eqref}[1]{(\ref{eq:#1})}

\def\l{\langle}

\def\r{\rangle}

\DeclareMathOperator{\diag}{diag}

\DeclareMathOperator{\Span}{span}
\DeclareMathOperator{\Rank}{rank}

\newcommand{\Kiprod}[2]{\langle \fmap(#1),\fmap(#2) \rangle_{\fspace}}
\newcommand{\Obs}{\mathcal{O}}
\newcommand{\otime}{L}
\newcommand{\JorMul}{O}
\newcommand{\Tset}{\Upsilon}

\newcommand{\Ind}{\mathcal{I}}

\newcommand{\JorP}{P}
\newcommand{\JorLa}{\Lambda}

\newcommand{\empKMod}{\widetilde{\empK}}
\newcommand{\empKBlock}{\widehat{\mathbf{K}}}
\newcommand{\ABlock}{\widehat{\mathbf{A}}}

\newcommand{\MBlock}{M}

\newcommand{\ITBlock}{I_2}
\newcommand{\ZBlock}{0}

\newcommand{\fspaceC}{\ensuremath{\fspace_{\centers}}}
\newcommand{\fspaceD}{\ensuremath{\fspace_{\centerscontrol}}}
\newcommand{\Atoms}{\mathcal{F}_{\centers}}
\newcommand{\AtomsControl}{\mathcal{F}_{\centerscontrol}}
\newcommand{\minmeas}{l}

\newcommand{\weightmap}{\ensuremath{\mathcal{W}}}
\newcommand{\weightmapI}{\ensuremath{\mathcal{W}^{-1}}}
\newcommand{\weightmapC}{\ensuremath{\acute{\mathcal{W}}}}

\newcommand{\ncontrol}{\ell}
\newcommand{\nevals}{r}
\newcommand{\evalvec}{\boldsymbol{\eval}}
\newcommand{\geomMult}{\gamma}
\newcommand{\weightc}{\acute{\weight}}
\newcommand{\empKD}{\empK_{\centerscontrol}}
\newcommand{\empKCD}{\empK_{\centers\centerscontrol}}
\newcommand{\empKShadFull}{\mathbf{\empK}}
\newcommand{\ControlMat}{\Psi}

\newcommand{\controlCent}{\mathcal{D}}
\newcommand{\controlCentLong}{\{d_1,\dots,d_{\ncontrol}\}}

\title{
Kernel Controllers: A Systems-Theoretic Approach for Data-Driven Modeling and Control of Spatiotemporally Evolving Processes 
}
\author{Hassan A. Kingravi, Harshal Maske and Girish Chowdhary
\thanks{ This work was supported in parts by DOE Award Number DE-FE0012173 and AFOSR Award Number 
FA9550-14-1-0399. Hassan Kingravi is with Pindrop Security, Harshal Maske, and Girish Chowdhary are with the Distributed Autonomous Systems (DAS) laboratory Oklahoma State University,{\tt \{hkingravi@pindropsecurity.com,  maske@okstate.edu, girish.chowdhary@okstate.edu\}}}%
}

\begin{document}

\maketitle

\begin{abstract}
We consider the problem of modeling, estimating, and controlling the latent state of a spatiotemporally evolving continuous function using very few sensor measurements and actuator locations. Our solution to the problem consists of two parts: a predictive model of functional evolution, and feedback based estimator and controllers that can robustly recover the state of the model and drive it to a desired function. We show that layering a dynamical systems prior over temporal evolution of weights of a kernel model is a valid approach to spatiotemporal modeling that leads to systems theoretic, control-usable, predictive models. We provide sufficient conditions on the number of sensors and actuators required to guarantee observability and controllability. The approach is validated on a large real dataset, and in simulation for the control of spatiotemporally evolving function.
\end{abstract}

\section{Introduction} \label{sec:intro}

Modeling, control, and estimation of spatiotemporally varying systems is a challenging area in controls research. These systems are characterized by dynamic evolution in both the spatial and temporal variables. Some examples of relevant problems include active wing-shaping based control of flexible aircraft, control of heat or particulate diffusion in manufacturing processes, control of rumor spreading across a social network, and tactical asset allocation and control problems in dynamically varying battlespaces. The traditional approach to modeling and control of spatiotemporal systems have relied on Partial Differential Equations (PDEs) \cite{Alabau2012controlpde}, solutions to which are functions that evolve in both space and time. However, PDE models can be limited in situations where exact physics based models of the functional evolution are difficult to formulate, or are inherently limited due to the physical understanding of the process or unknown spatiotemporal interactions \cite{cressie2011statistics}. Furthermore, the control of PDEs is fundamentally more challenging than the control of finite-dimensional state-space systems because the evolution and control spaces are infinite dimensional Hilbert spaces, as opposed to $\mathbb{R}^n$ \cite{Alabau2012controlpde}.

Accordingly, there has been significant work in approximate modeling of spatiotemporally evolving functions using data-driven or distributed parameter based approximations of PDEs \cite{cressie2011statistics,wikle2001kernel}.   
One way to model spatiotemporally evolving functions is to approximate the function at several sampling locations and build an autoregressive model of the evolution of the function's output over that grid \cite{baker2000finite}. The fidelity of these models heavily depends on the number of sampling (equivalently Euclidean grid locations in the independent variable space) locations employed, with a large number of grid locations leading to large-scale state-space models that are difficult to manage. An alternative approach to modeling spatiotemporal functional evolution relies on modeling the correlation between any two sampling locations through a smooth covariance kernel \cite{cressie2011statistics}. The model of the evolution is then formed through a linear, weighted combination of the kernels, and the hyperparameters of the spatiotemporal covariance kernel and the weights are learned by solving an optimization problem. The power and flexibility of this approach lies in the fact that kernels can be defined over abstract objects, and not just Euclidean grid locations, leading to a modeling technique that is domain agnostic. For example, kernel embeddings are available for graphical models studied in decentralized control \cite{johansson2014global}, images \cite{ren2012coupled}, and many other domains. However, formulating control-usable kernel-based models of spatiotemporal phenomena can be challenging due to the need to take into account the spatiotemporal dependence. Many recent techniques in spatiotemporal modeling have focused on covariance kernel design and associated hyperparameter learning algorithms \cite{garg2012AAAI,ma2003nonstationary,RasmussenWilliams2005,plagemann2008nonstationary}. The main benefit of careful design of covariance kernels over approaches that simply include time in as an additional input variable \cite{perez:13:gaussian,Chowdhary13_ACC2} is that they can account for intricate spatiotemopral couplings. However, there are two key challenges with these approaches: the first challenge is in ensuring the scalability of the model to large scale phenomena. This is difficult due to the fact that the hyperparameter optimization problem is not convex in general, and because when time is used as a kernel input, it is nontrivial to restrict the number of kernels used without losing modeling fidelity \cite{garg2012AAAI,ma2003nonstationary,plagemann2008nonstationary}. 
The second very important challenge is concerned with the formulation of feasible control strategies utilizing predictive kernel-based models of spatiotemporal phenomena. In particular, when the spatiotemporal evolution is embedded in the design of complex covariance kernel, the resulting model of functional evolution can be highly nonlinear and difficult to utilize in control design. 

In this paper, we pursue an alternative systems-theoretic approach to the modeling, control, and estimation of spatiotemporally varying functions that fuses the strengths of kernel methods with systems theory. Our main contribution is to provide a systems-theoretic formulation for approximating, with very high accuracy, spatiotemporal functional evolution by layering a linear dynamical systems prior over temporal evolution of weights of a kernel model. For a class of linearly evolving PDEs, such as the heat diffusion and the wave equation, our approach can lead to a very high-accuracy approximation. This modeling approach is also applicable to data-driven modeling of real-world phenomena, which we demonstrate on a challenging inference problem on satellite data of sea surface temperatures. 
One benefit of our model is that it can encode spatiotemporal evolution of complex nonlinear surfaces through an Ordinary Differential Equation (ODE) evolving in a Hilbert space induced by the specific kernel choice.
Yet, the main benefit of our systems-theoretic approach is that it is highly conducive to control synthesis. To illustrate this fact, we demonstrate that feasible control strategies for a class of spatiotemporally evolving systems can be found using linear control synthesis. In particular, we derive sufficient conditions on the kernel selection to guarantee observability and controllability of the presented model. Furthermore, we demonstrate control synthesis for a diffusion PDE using simple Gaussian kernels distributed uniformly in the input domain.  

The outline of this paper is as follows, Section \ref{sec:observers} focuses on the development of a systems-theoretic kernel-based model of spatiotemporal evolution, Section \ref{sec:theory_results} presents the main theoretical results, Section \ref{sec:experimental_results} presents modeling results on a real-world large dataset and control synthesis results for a diffusion PDE.

\section{Kernel Controllers}\label{sec:observers}
This section outlines our modeling framework and presents theoretical results associated with the number of sampling locations required for monitoring functional evolution. 

\subsection{Problem Formulation}\label{sec:formulation}
We focus on predictive inference and control over a time-varying stochastic process, whose mean $f$ is temporally evolving: 
\begin{eqnarray}\eqlabel{e:bnp_model1}
& f_{k+1}& \sim \mathbb{F}(f_k,\eta_k) 
\end{eqnarray}
where 
$\mathbb{F}$ is a distribution varying with time $t$ and exogenous inputs $\eta$. 
The theory of reproducing kernel Hilbert spaces (RKHSs) provides powerful tools for generating flexible classes of functions with relative ease, and is thus a natural choice for modeling complex spatial functions  \cite{scholkopf:bk:2002}. Therefore, our focus will be on spatiotemporally evolving kernel-based models, such as Gaussian Processes (GPs). 
In a kernel-based model,
$\kernel:\dom\times\dom\to\R$ is a positive definite kernel on some compact domain $\dom$ that models the covariance between any two points in the input space. A Mercer kernel \cite{scholkopf:bk:2002} implies the existence of a smooth map $\fmap:\dom\to\fspace$, where $\fspace$ is an RKHS with the property
\begin{align}
 \kernel(x,y) &= \Kiprod{x}{y} = \Kiprod{\kernel(x,\cdot)}{\kernel(y,\cdot)}. 
\end{align}
There is a large body of literature  on modeling spatiotemporal evolution in $\fspace$ \cite{wikle2002kernel,cressie2011statistics}. 
A simple approach for spatiotemporal modeling is to utilize both spatial and temporal variables as inputs to the kernel \cite{perez:13:gaussian,Chowdhary13_ACC2}. However, this technique leads to an ever-growing kernel dictionary, which is computationally taxing.
Furthermore, constraining the dictionary size or utilizing a moving window will 
occlude the learning of long-term patterns. Periodic or nonstationary covariance functions and nonlinear transformations have been proposed to address this issue \cite{ma2003nonstationary,RasmussenWilliams2005}. Furthermore, work in the design of nonseparable and nonstationary covariance kernels seeks to design kernels optimized to environment-specific dynamics, and optimize their hyperparameters in local regions of the input space  \cite{garg2012AAAI,das2014nonstationary,plagemann2008nonstationary}. 
The model of spatiotemporal functional evolution proposed in this paper builds on the idea that modeling the temporal evolution of mixing weights of a kernel model is a valid approach to spatiotemporal modeling. The key idea behind our approach is that the spatiotemporal evolution of a kernel-based model can be directly modeled by tracing the evolution of the mean embedded in a RKHS using switched ordinary differential equations (ODE) when the evolution is continuous, or switched difference equations when it is discrete (Figure \ref{fig:hilbert_evolution}).
The advantage of this approach is that it allows us to utilize powerful ideas from systems theory for knowing necessary conditions for functional convergence; furthermore, it offers a natural framework for designing control mechanisms as well. 
\begin{figure}[t]
\centering
 \includegraphics[width=0.8\columnwidth]{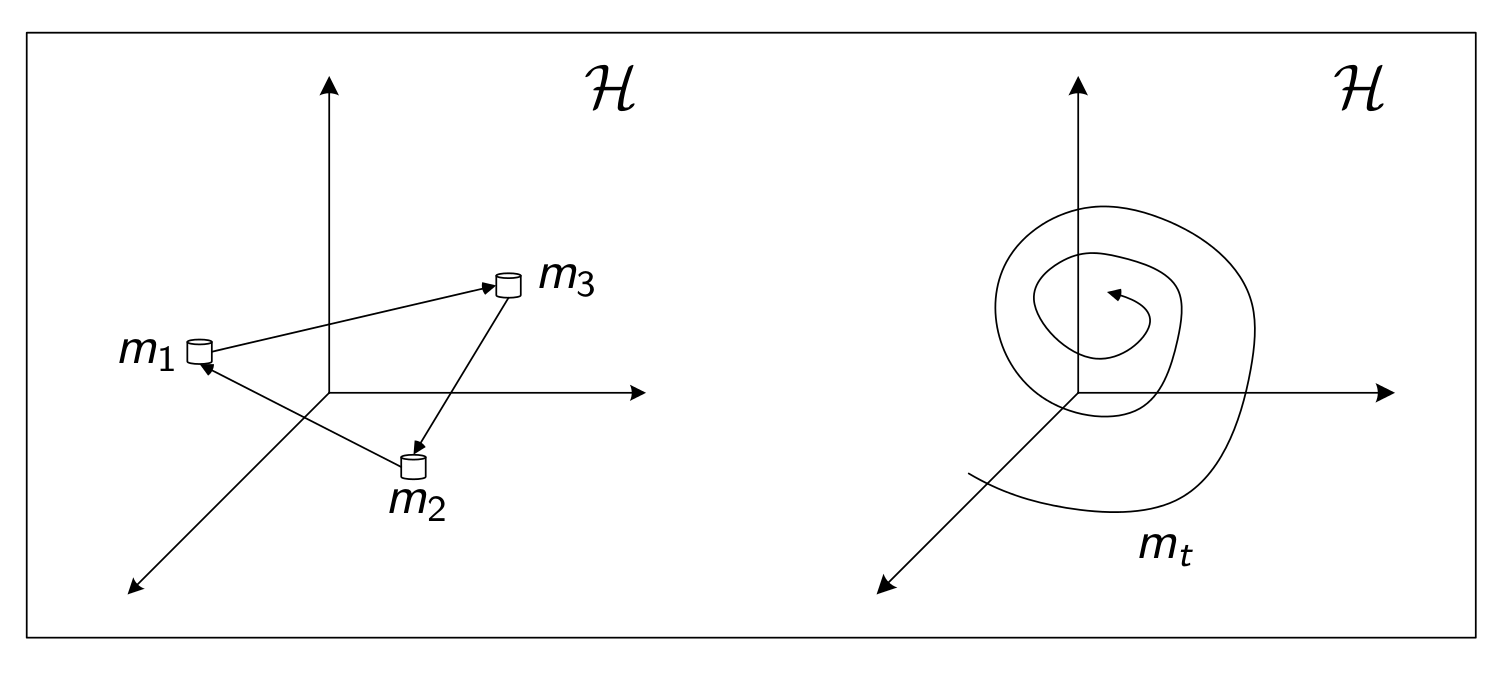}
 \vspace{-0.1in}
 \caption{Two types of Hilbert space evolutions. Left: the model, represented by the 
          functions $m_i$, switches discretely in the Hilbert space $\fspace$; Right: 
          the evolution of the function $m_t$ is smooth, represented by a solution to an 
          ordinary differential equation in $\fspace$.}
\label{fig:hilbert_evolution}
\vspace{-0.2in}
\end{figure}
In this paper, we restrict our attention to the class of functional evolutions $\mathbb{F}$ defined by linear Markovian transitions in an RKHS. While extension to the nonlinear case is possible (and non-trivial), it is not pursued in this paper to help ease the exposition of key ideas. 
Let $y\in\R^{\nsamp}$ be the measurements of the function available from $\nsamp$ sensors, $\sysop:\fspace\to\fspace$ be a linear transition operator in the RKHS $\fspace$, and $\measop:\fspace\to\R^{\nsamp}$ be a linear measurement operator, the model for the infinite-dimensional functional evolution and measurement studied in this paper is:
\begin{align}\eqlabel{ideal_lin_evol}
 f_{k+1} &= \sysop f_k + \eta_k\\
 y_k &= \measop f_k + \zeta_k,
\end{align}
where $\eta_k$ is a zero-mean stochastic process in $\fspace$, and $\zeta_k$ is a Wiener process in $\R^{\nsamp}$. 
For many kernels, the feature map $\fmap$ is unknown, and therefore it is necessary to work in the dual space of $\fspace$. For concreteness, we work with an approximate space as follows: given
points $\shCent = \shCentLong$, $c_i\in\dom$, we have a dictionary of atoms $\Atoms = \begin{bmatrix}\fmap(c_1) &\cdots & \fmap(c_{\ncent})
\end{bmatrix}$, $\fmap(c_i)\in\fspace$, the span of which is a strict subspace of the RKHS generated by the kernel. 
Formally, we have
\begin{align}
 \shCent\mapsto \fspaceC := \Span 
 \begin{bmatrix}\fmap(c_1) &\cdots & \fmap(c_{\ncent})
\end{bmatrix}
\subset \fspace.
\end{align} 
This regime, which trades off the flexibility of a truly nonparametric approach for computational realizability, still allows for the representation of rich phenomena. 
Let $\nsamp$ represent the number of sampling locations, and $\ncent$ be the number of bases generating 
$\fspaceC$. Note that every function $f\in\fspaceC$ has an expansion of the form 
\begin{align}\eqlabel{expansion}
 f(x) = \sum_{i=1}^{\ncent} w_i k(c_i,x).
\end{align}
This expansion allows us to write the $w_i$ coordinates in the dual space as vectors $w\in\R^{\ncent}$.  
We can show the relation of the function spaces to their Euclidean counterparts via commutative diagrams. Define $\weightmap:\fspaceC\to\R^{\ncent}$ as the operator that maps the coordinates $w_i$ in \eqref{expansion} to vectors $\weight\in\R^{\ncent}$, and let $\weightmapI:\R^{\ncent}\to\fspaceC$. Note that for finite-dimensional spaces, this inverse map always exists. These definitions allow us to outline the relations between the dynamics operators $\sysop$ and $A$, and the measurement operators $\measop$ and $\empK$ using the commutative diagrams in Figure \ref{fig:commute_sysop_dyn} and Figure \ref{fig:commute_sysop_meas} respectively. 
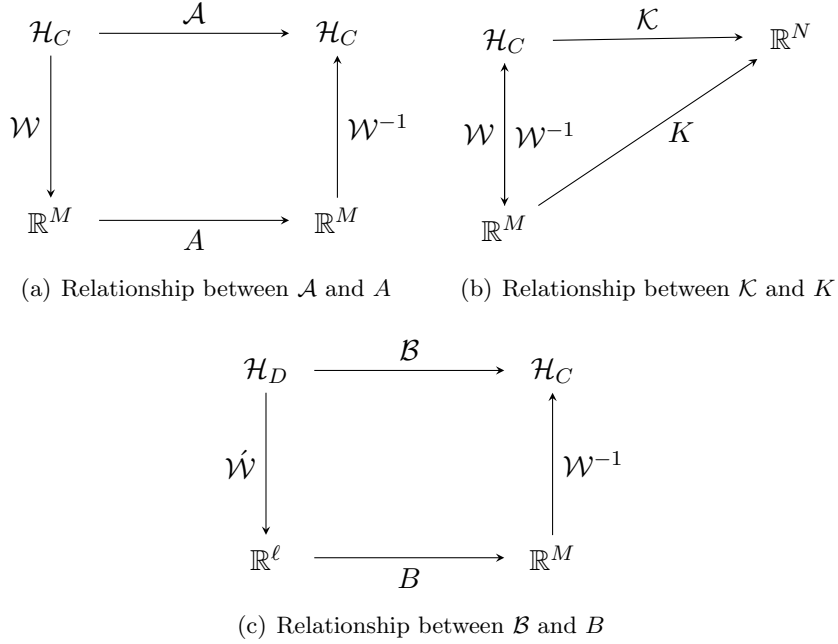
\begin{figure}
\centering
\subfigure[Relationship between $\sysop$ and $A$]{
\label{fig:commute_sysop_dyn}
\begin{tikzpicture}
  \matrix (m) [ampersand replacement=\&, matrix of math nodes, 
                 row sep=0.75in,column sep=1in,minimum width=0.5in] {
     \fspaceC \& \fspaceC \\
     \R^{\ncent} \& \R^{\ncent} \\};
  \path[-stealth]
    (m-1-1) edge node [left] {$\weightmap$} (m-2-1)
            edge [right] node [above] {$\sysop$} (m-1-2)
    (m-2-1.east|-m-2-2) edge node [below] {$ A$} node [above] {} (m-2-2)
    (m-2-2) edge node [right] {$\weightmapI$} (m-1-2);            
\end{tikzpicture}  
}
\
\subfigure[Relationship between $\measop$ and $\empK$]{
\label{fig:commute_sysop_meas}
\begin{tikzpicture}
  \matrix (m) [ampersand replacement=\&, matrix of math nodes, 
               row sep=0.75in,column sep=1in,minimum width=0.5in] {
     \fspaceC \& \R^{\nsamp} \\
     \R^{\ncent} \&  \\};
  \path[-stealth]
    (m-1-1) edge node [left] {$\weightmap$} (m-2-1)
            edge [right] node [above] {$\measop$} (m-1-2)
    (m-2-1) edge node [right] {$\ \empK$} (m-1-2)
    (m-2-1) edge node [right] {$\weightmapI$} (m-1-1);            
\end{tikzpicture}
}
\
\subfigure[Relationship between $\controlop$ and $B$]{
\label{fig:commute_sysop_control}
\begin{tikzpicture}
  \matrix (m) [ampersand replacement=\&, matrix of math nodes, 
                 row sep=0.75in,column sep=1in,minimum width=0.5in] {
     \fspaceD \& \fspaceC \\
     \R^{\ncontrol} \& \R^{\ncent} \\};
  \path[-stealth]
    (m-1-1) edge node [left] {$\weightmapC$} (m-2-1)
            edge [right] node [above] {$\controlop$} (m-1-2)
    (m-2-1.east|-m-2-2) edge node [below] {$ B$} node [above] {} (m-2-2)
    (m-2-2) edge node [right] {$\weightmapI$} (m-1-2);            
\end{tikzpicture}
}
\caption{Commutative diagrams between primal and dual spaces}
\label{fig:commute_sysop}
\end{figure}
The finite-dimensional evolution equations equivalent to \eqref{ideal_lin_evol} in the dual space can be formulated as
\begin{align} 
 \weight_{k+1} &= A\weight_k + \eta_k \eqlabel{k_measure}\\
 y_{k} &= \empK_{k} w_{k} +\zeta_k, \eqlabel{k_measure1}
\end{align}
where we have matrices $A\in \R^{\ncent\times\ncent}, \ \empK_{k}\in \R^{\nsamp\times\ncent}$, the vectors $\weight_k, \weight\in\R^{\ncent}$, and we have slightly abused notation to let $\eta_k$ and $\zeta_k$ denote their $\fspaceC$ counterparts.
Note that the measurement operator $\measop$ is simply a sampling of the function $f$ at an arbitrary set of sensing locations $\sampSet = \sampSetLong$, where $x_i\in\dom$: we will see how this affects the structure of $\empK_k$ momentarily. 

The equations \eqref{ideal_lin_evol} suggest an immediate extension to functional control problems. 
Pick another dictionary of atoms $\AtomsControl = \begin{bmatrix}\fmap(d_1) &\cdots & \fmap(d_{\ncontrol})
\end{bmatrix}$, $\fmap(d_j)\in\fspace$, $d_j\in\dom$, the span of which, denoted by $\fspaceD$, is a strict subspace of the RKHS $\fspace$  generated by the kernel. 
The functional evolution equation is then as follows:
\begin{align}\eqlabel{ideal_lin_evol_control}
 f_{k+1} &= \sysop f_k + \controlop \control_k + \eta_k\\
 y_k &= \measop_k f_k + \zeta_k,
\end{align}
where the control functions $\control_k$ evolve in $\fspaceD$, and $\controlop:\fspaceD\to\fspaceC$. To derive the finite-dimensional equivalent of $\controlop$, we have to work out the structure of the matrix $B$: since $\fspaceC$ is not, in general, isomorphic to $\fspaceD$, this imposes strict restrictions on $B$. We derive $B$ using least squares using the inner product of $\fspace$. Let $\control = \sum_{j=1}^{\ncontrol}\weightc_j\kernel(d_j,x)$, and let $\Atoms = 
 \begin{bmatrix}\fmap(c_1) &\cdots & \fmap(c_{\ncent})
\end{bmatrix}$ be the basis for $\fspaceC$. Then the projection of $\delta$ onto $\fspaceC$ can be derived as 
\begin{align*}
 \begin{bmatrix}
  \l \delta, \fmap(c_1) \r_{\fspace}\\
   \vdots\\
  \l \delta, \fmap(c_{\ncent}) \r_{\fspace} 
 \end{bmatrix}
 &= 
 \underbrace{
 \begin{bmatrix}
  \kernel(d_1,c_1) & \cdots & \kernel(d_{\ncontrol},c_1)\\
   \vdots  &\ddots &\vdots\\
  \kernel(d_1,c_{\ncent}) & \cdots & \kernel(d_{\ncontrol},c_{\ncent})
 \end{bmatrix}}_{\empKCD}
 \begin{bmatrix}
  \weightc_1\\
  \vdots\\
  \weightc_{\ncontrol}
 \end{bmatrix},
\end{align*}
using the reproducing property. 
This derivation shows that the operator $B = \empKCD\in\R^{\ncent\times\ncontrol}$, the kernel matrix between the data $\centers$ generating the atoms $\Atoms$ of $\fspaceC$ and the data $\centerscontrol$ generating the atoms $\AtomsControl$ of $\fspaceD$. Using similar arguments, it can be shown that, given sensing locations $\data = \{x_1,x_2,\dots, x_{\nsamp}\}$, $\empKD\in\R^{\nsamp\times\ncontrol}$ is the kernel matrix between $\data$ and $\centerscontrol$. Thus the finite-dimensional evolution equations equivalent to \eqref{ideal_lin_evol_control} are
\begin{align}
 \weight_{k} &= A\weight_k + \empKCD\weightc_k\eqlabel{k_measure_c}\\
 y_{k} &= \empK_{k} w_{k} \eqlabel{k_measure_c1}.
\end{align}
We pause here to point out just how flexible the kernel-based framework is. First of all, the choice of kernel completely determines the space $\fspace$, which may allow wildly different functional outputs for the same dynamics matrix, as shown in Figure \ref{fig:kernel_variation}. Note also that the dynamical equations \eqref{k_measure_c} and \eqref{k_measure_c1} are \emph{independent of the choice of domain $\dom$}: different domains with different kernels may result in the same sequence of matrices $\empK_k$. This allows our results to hold for any domain over which a kernel can be defined, including examples like graphs, hidden Markov models, and strings, which are not typically studied in the controls literature, at virtually no extra complexity in implementation beyond the design of the actual sensors and actuators. This remarkable fact is why we denote our method to be \emph{domain agnostic}. 
\begin{figure}[t]
\centering
\subfigure[Gaussian]{
  \includegraphics[width=0.35\columnwidth]{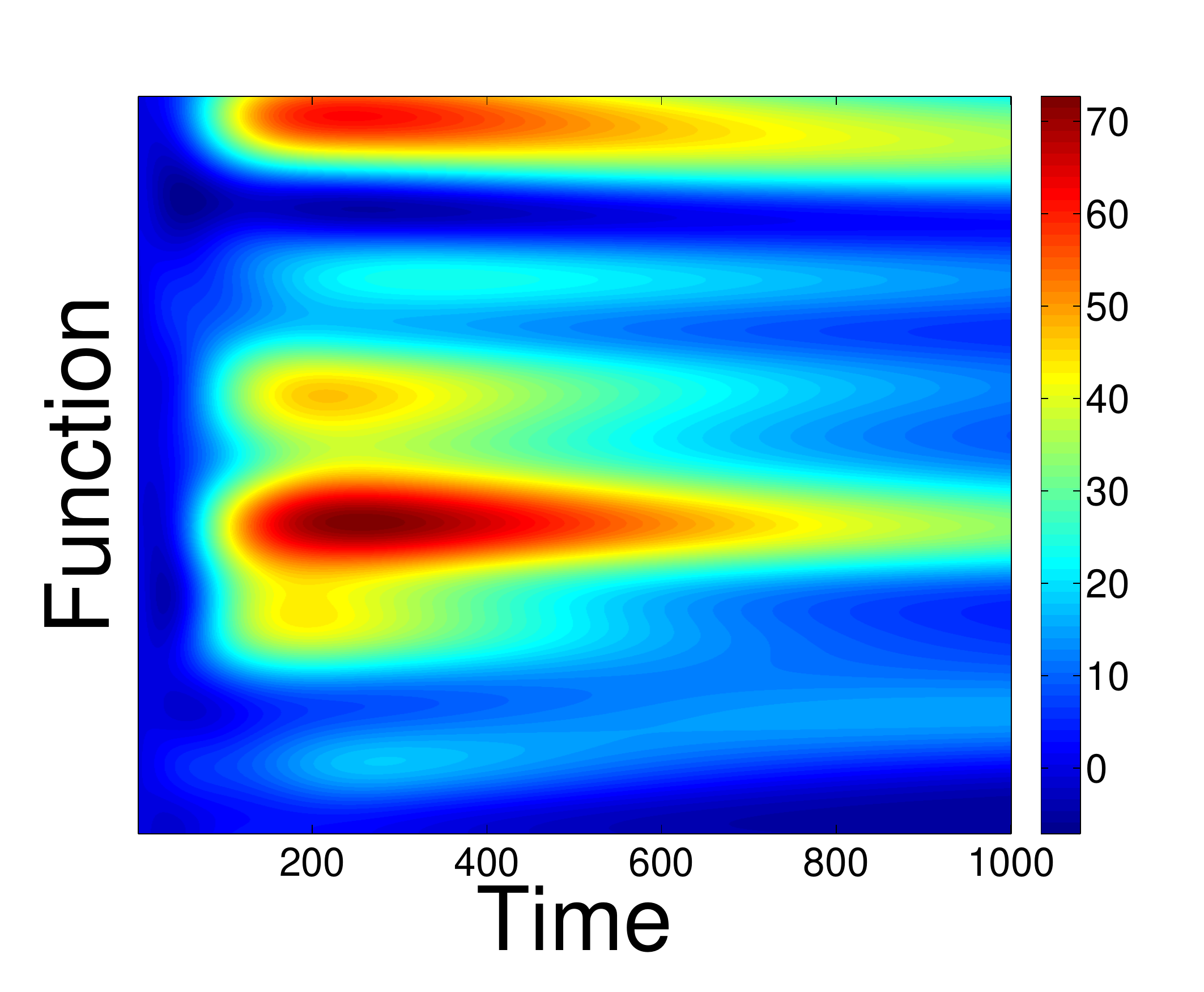}
}
\subfigure[Laplacian]{
  \includegraphics[width=0.35\columnwidth]{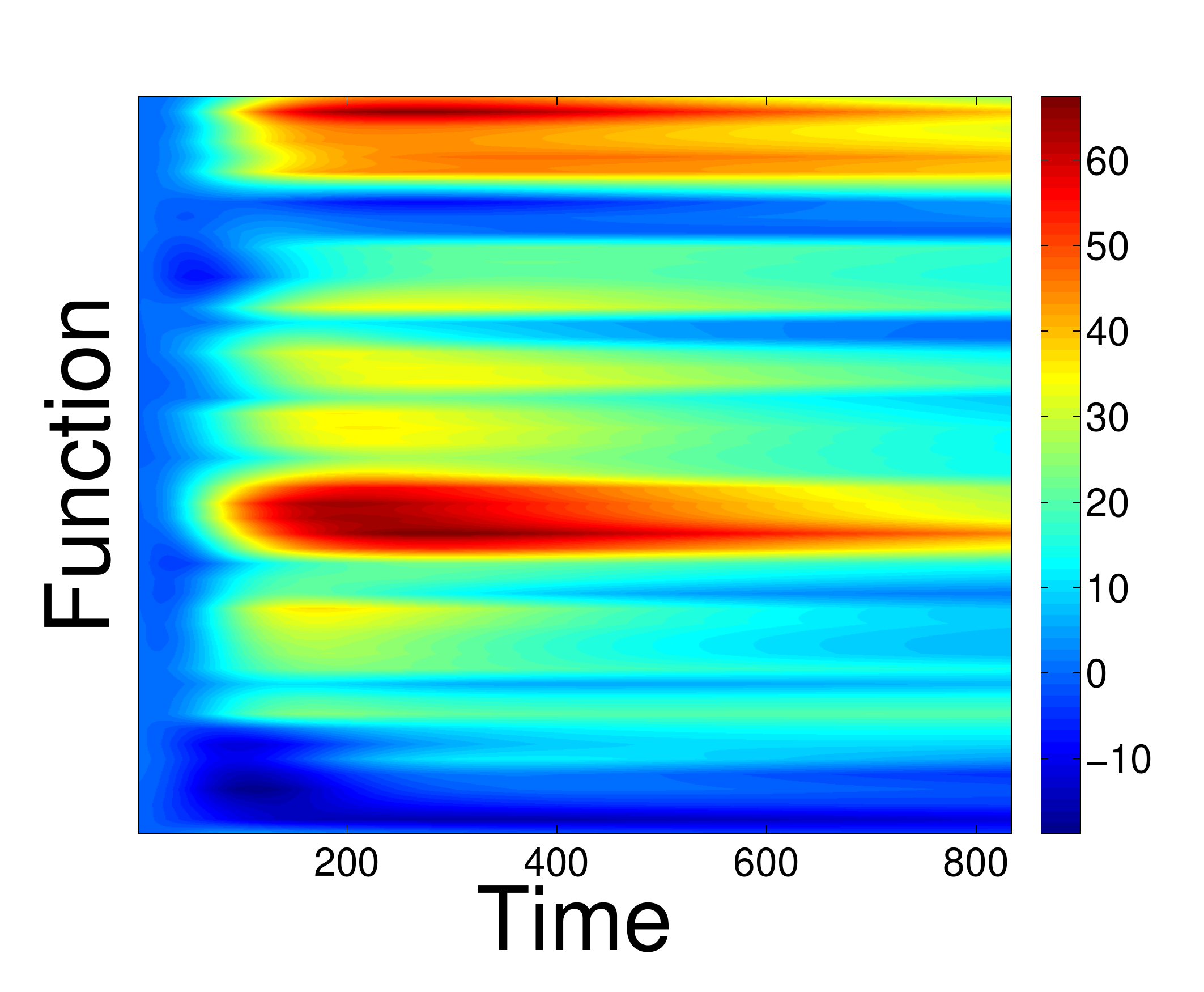}
}
\\
\subfigure[Periodic]{
  \includegraphics[width=0.35\columnwidth]{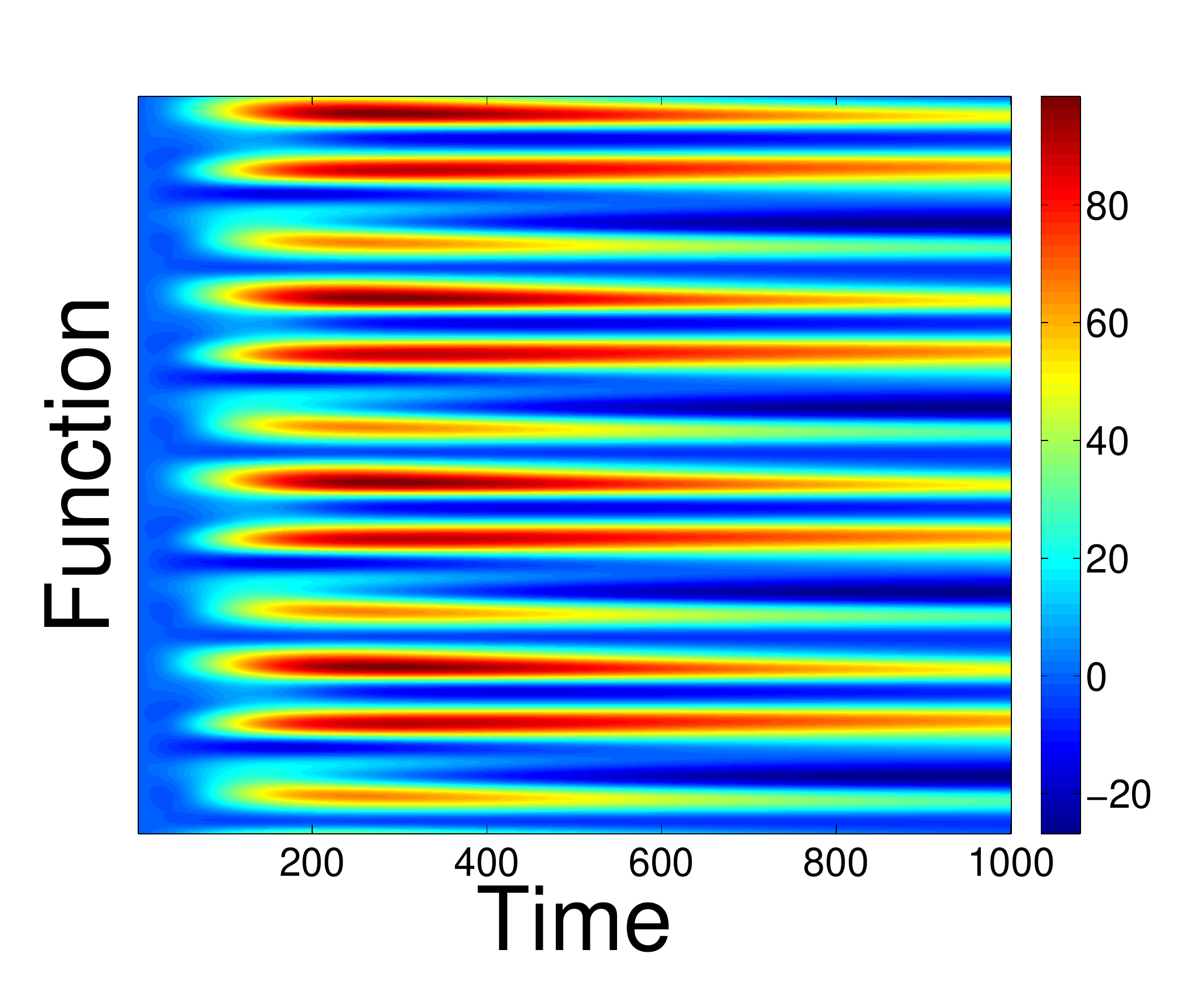}
}
\subfigure[Locally periodic]{
  \includegraphics[width=0.35\columnwidth]{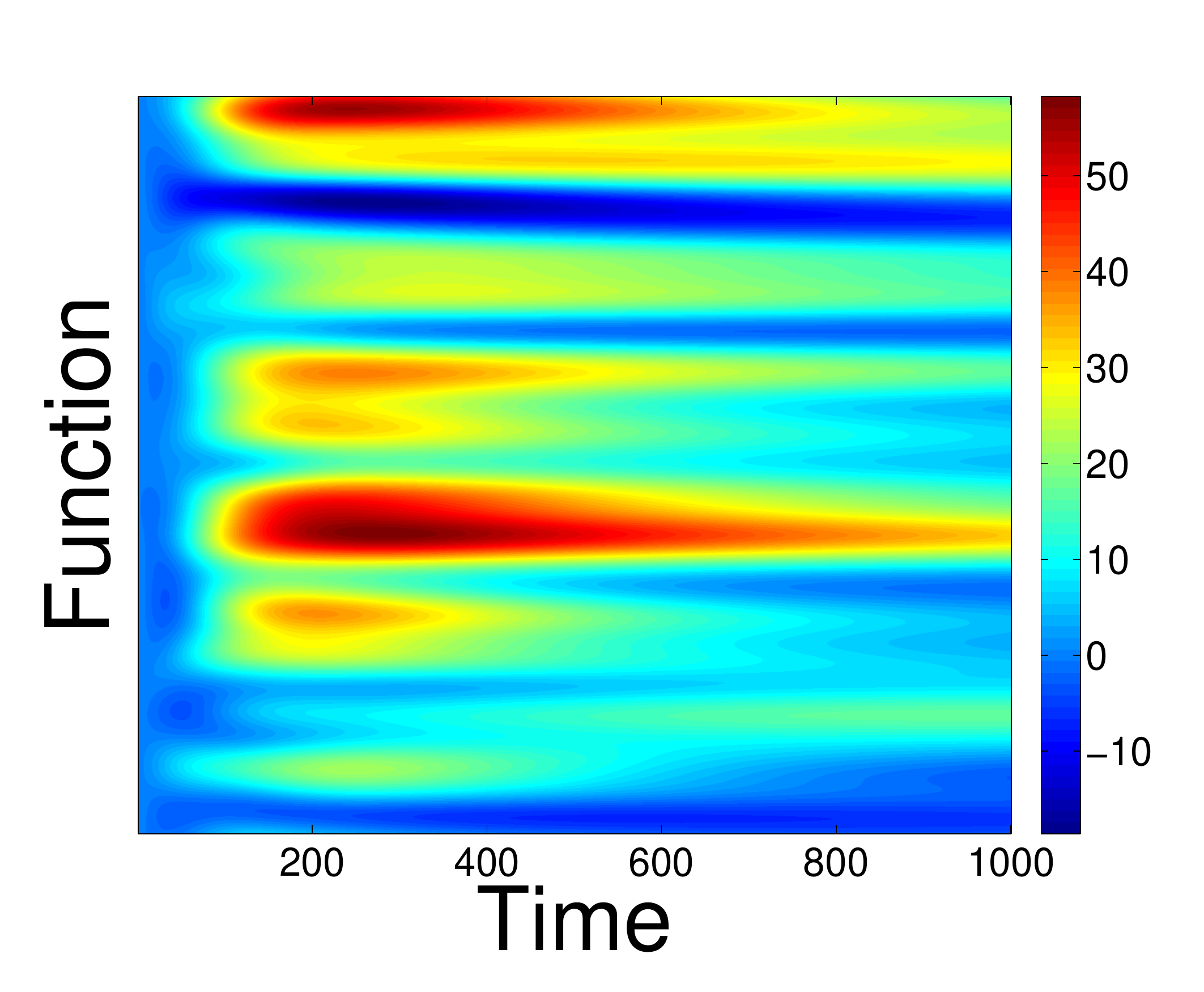}
}
\vspace{-0.1in}
\caption{One-dimensional function evolution over a fixed systems matrix $A$, 
initial condition $\weight_0$ and centers $\shCent$, but with different kernels $\kernel(x,y)$. 
Each $y$-vector at a given value of $x$ represents the output of the function 
which evolves from left to right. 
As can be seen, changing the kernel creates quite different behavior for the same system. 
}
\label{fig:kernel_variation}
\vspace{-0.2in}
\end{figure}

Since $\empK_{k+1}$ is the kernel matrix between the data points and basis vectors, its rows are of the form 
 $\empK_{(i)} = 
 \begin{bmatrix}
  \kernel(x_i, c_1) & \kernel(x_i, c_2) & \cdots & \kernel(x_i, c_{\ncent}) 
 \end{bmatrix}$. 
In systems-theoretic language, each row of the kernel matrix corresponds to a \emph{measurement} at a particular location, and the matrix itself acts as a measurement operator. We define the \emph{generalized observability matrix} \cite{zhou:bk:96} as 
\begin{align}\eqlabel{obs_mat}
 \Obs_{\Tset} = 
 \begin{bmatrix}
  \empK_{t_1} A^{t_1}\\  
  \cdots\\
  \empK_{t_\otime} A^{t_\otime}
 \end{bmatrix},
\end{align}
where $\Tset = \{t_1, t_2, \dots, t_{\otime}\}$ are the set of instances $t_i$
when we apply the measurement operators $\empK_{t_i}$. 
Note that $\Obs_{\Tset}\in\R^{\nsamp\otime\times\ncent}$. 
Similarly, we can define the \emph{generalized controllability matrix} as 
\begin{align}\eqlabel{control_mat}
 \ControlMat_{\Tset} = 
 \begin{bmatrix}
   {A^{t_1}}^T {\empKD}_{t_1} & {A^{t_2}}^T {\empKD}_{t_2} & \cdots {A^{t_\otime}}^T {\empKD}_{t_\otime}
 \end{bmatrix},
\end{align}
$\ControlMat_{\Tset}\in\R^{\ncent\times\otime\ncontrol}$
A linear system is said to be observable if $\Obs_{\Tset}$ has full column rank (i.e. $\mathrm{Rank}~\Obs_{\Tset}=\ncent$)
and is controllable if $\ControlMat_{\Tset}$ has full row rank, for $\Tset = \{0, 1, \dots, \ncent-1\}$ \cite{zhou:bk:96}. 

Observability guarantees that a feedback-based observer can be designed such that the estimate of $w$ denoted by $\hat {w_k}$ converges exponentially fast to the true state $w_k$. %
In particular, observability is the necessary condition for the existence of a unique solution to the Riccatti equation required in designing a Kalman filter. Therefore, when $\eta,\zeta$ have a zero mean Gaussian distribution, a Bayes optimal filter can be designed for estimating $w$ if and only if $\mathrm{Rank}~\Obs_{\Tset}=\ncent$. 
Similarly, controllability guarantees that a feedback-based controller can drive the current functional state of the system $f_{k}$ to a reference function $f_{\text{ref}}$, as long as $f_{\text{ref}}\in\fspaceC$.

We are now in a position to formally state the spatiotemporal monitoring and control problem considered: Given a spatiotemporally evolving system modeled using \eqref{ideal_lin_evol_control}, choose a set of $\nsamp$ sensing locations $\sampSet = \sampSetLong$ and $\ncontrol$ actuating locations $\controlCent = \controlCentLong$
such that even with $\nsamp\ll \ncent$ and $\ncontrol\ll \ncent$, the functional evolution of the spatiotemporal model can be estimated robustly, and driven (controlled) to a reference function $f_{\text{ref}}$. Our approach to solve this problem relies on the design of the measurement operator  $\empK$ such that the pair $(A,\empK)$ is observable, and the control operator $\empKD$ such that the pair $(A,\empKD)$ is controllable. 

\subsection{Theoretical Results}\label{sec:theory_results}
In this section, we prove results concerning the observability of spatiotemporally varying functions modeled by the functional evolution and measurement equations \eqref{k_measure} and \eqref{k_measure1} formulated in Section \ref{sec:formulation}. In particular,  observability of the system states implies that we can recover the current state of the spatiotemporally varying function using a small number of sampling locations $\nsamp$, which allows us to 1) track the function, and 2) predict its evolution forward in time. It should be noted that the results are also applicable to controllability of the system in \eqref{k_measure_c1} since the structure of the control matrix $K_{CD}$ is also that of a Kernel matrix.  We first show in Proposition \ref{prop:1} that if $A$ has a full-rank Jordan decomposition, the kernel matrix meeting a  condition called \emph{shadedness} (to be defined below) is sufficient for the system to be observable. In Proposition \ref{prop:2}, we prove a lower bound on the number of sampling locations required for observability which holds for more general $A$. Finally, in Proposition \ref{prop:3}, we outline a method that achieves this lower bound for certain kernels. Since both $\empK$ and $\empKCD$ are kernel matrices generated from a shared kernel, these observability results translate directly into controllability results. 

To prove our results, we will leverage the spectral decomposition of $A$. Specifically, recall that any matrix $A\in\R^{\ncent\times\ncent}$ is similar to a unique block diagonal matrix $\JorLa$ (i.e. $\exists P\in\R^{\ncent\times\ncent}$ invertible such that $A = P\JorLa P^{-1}$) whose diagonal blocks are matrices of the form 
\begin{align}
\JorLa_k(\lambda_i, \lambda_i^*) := 
\begin{bmatrix}
\MBlock & \ITBlock & \cdots &  \ZBlock\\
\vdots & \vdots & \ddots  & \ITBlock\\
\ZBlock & \ZBlock & \cdots & \MBlock \eqlabel{jor_com}
\end{bmatrix}.
\end{align}
where  $(\lambda_i, \lambda_i^*)$ is a complex conjugate eigenvalue of $A$, and 
$
\MBlock = \begin{bmatrix}
\mu_1 & \mu_2\\
-\mu_2 & \mu_1
\end{bmatrix}$ and $ \ITBlock= \begin{bmatrix}
1 & 0\\
0 & 1
\end{bmatrix}$. 
Real eigenvalues $\la_i$ correspond to the case $\MBlock = \la_i$ and $\ITBlock = 1$. 
Thus the complete real Jordan form of $A$ will be the appropriate diagonal array of these blocks. If all the eigenvalues $\lambda_i$ are nonzero and real, we say the matrix has a \emph{full-rank Jordan decomposition}.

\begin{definition}\label{def:shaded}
	\textbf{(Shaded Kernel Matrix)} Let $\kernel:\dom\times\dom\to\R$ be a positive-definite kernel on a compact domain $\dom$. 
	Let $C = [c_1,  c_2, \cdots , c_{\ncent}\}$, $c_j\in\dom$ be the points generating a 
	finite-dimensional covering of the reproducing kernel Hilbert space $\fspace$ associated to $\kernel(x,y)$,
	and let $\sampSet = \sampSetLong$, $x_i\in\dom$ 
	Let $\empK\in\R^{\nsamp\times\ncent}$ be the kernel matrix, where $ \empK_{ij} := \kernel(x_i,c_j)$. For each
	row $\empK_{(i)} := [k(x_i,c_1), k(x_i,c_2),\dots, k(x_i,c_{\ncent})]$, define the set 
	$\Ind_{(i)} := \{\iota_1^{(i)},\iota_2^{(i)},\dots, \iota_{\ncent_i}^{(i)}\}$ to be the indices in the kernel
	matrix row $i$ which are nonzero. 
	Then if 
	\begin{align}\eqlabel{shaded_cond}
	\bigcup_{1 \leq i \leq \nsamp} \Ind^{(i)} = \{1,2,\dots, \ncent\},
	\end{align}
	we denote $\empK$ as a \emph{shaded kernel matrix} (see figure \ref{shaded_matrix}).
\end{definition}
This condition implies that the null space of the adjoint of $\empK$ as a linear operator between Euclidean spaces, i.e. $\empK^T:\R^{\nsamp}\to\R^{\ncent}$ is trivial. Note that, in principle, for the Gaussian kernel, a single row generates a shaded kernel matrix, although this matrix can have many entries that are extremely close to zero.  With this definition in place, we can prove the following proposition, which shows that if $A$ has a full-rank Jordan decomposition, a shaded kernel matrix is sufficient to prove observability. 


\begin{proposition}\label{prop:1}
	Let $\kernel:\dom\times\dom\to\R$ be a positive definite kernel on a domain $\dom$. 
	Let $C = [c_1,  c_2, \cdots , c_{\ncent}\}$, $c_j\in\dom$ be the points generating a 
	finite-dimensional covering of the reproducing kernel Hilbert space $\fspace$ associated to $\kernel(x,y)$, and
	consider the discrete linear system on $\fspace$ given by the evolution and measurement equations \eqref{k_measure} and \eqref{k_measure1}. 
	Let $A\in\R^{\ncent\times\ncent}$ be a full-rank Jordan decomposition of the form $A = \JorP\JorLa\JorP^{-1}$,
	where $\JorLa = \diag(\begin{bmatrix}\JorLa_1 & \JorLa_2 &\cdots & \JorLa_{\JorMul}\end{bmatrix})$,
	and there are no repeated eigenvalues. 
	Given a set of time instances  $\Tset = \{t_1,t_2,\dots,t_{\otime}\}$, and a set of sampling locations $\sampSet=\sampSetLong$,
	the system \eqref{k_measure} is observable if the kernel matrix $\empK_{ij} := \kernel(x_i,c_j)$ is shaded,
	$\empK^D$, the row vector generated by summing the rows of $\empK$, has all nonzero entries, 
	$\Tset$ has distinct values, and $|\Tset| \geq \ncent$.
\end{proposition}
\begin{proof}
 To begin, consider a system where $A = \JorLa$, with Jordan blocks $\{\JorLa_1, \JorLa_2, \dots, \JorLa_{\JorMul}\}$ along the 
 diagonal. Then $A^{t_i} = \diag(\begin{bmatrix}\JorLa_1^{t_i} & \JorLa_2^{t_i} & \cdots & \JorLa_{\JorMul}^{t_i}\end{bmatrix})$. 
 We have that   
 \begin{equation*}
 \Obs_{\Tset} = 
 \begin{bmatrix}
  \empK A^{t_1}\\
  \cdots\\
  \empK A^{t_\otime}
 \end{bmatrix}
 =
 \underbrace{
 \begin{bmatrix}
  \empK & \cdots & \empK
 \end{bmatrix}}_{\empKBlock\in\R^{\nsamp\times\ncent\otime}}
 \underbrace{
 \begin{bmatrix}
  \JorLa_1^{t_1} & \cdots & 0\\
   \vdots & \ddots & \vdots\\
  0 & \cdots & \JorLa_{\JorMul}^{t_1}\\ 
   \hline
   \vdots & \ddots & \vdots\\
   \hline
  \JorLa_1^{t_{\otime}} & \cdots & 0\\
   \vdots & \ddots & \vdots\\
  0 & \cdots & \JorLa_{\JorMul}^{t_\otime}   
 \end{bmatrix}}_{\ABlock\in\R^{\ncent\otime\times\ncent}}
\end{equation*}
Recall that a matrix's rank is preserved under a product with an invertible matrix. Design a 
matrix $U\in\R^{\nsamp\times\nsamp}$ s.t. $\empKMod := U\empK$ is a matrix with one row vector of nonzeros, and all of the 
remaining rows as zeros. Then $\Rank(\empKBlock\ABlock) = \Rank(U\empKBlock\ABlock)$. 
Therefore, we have that 
\begin{equation*}
 \empKMod A^{t_j} = 
 \begin{bmatrix}
  \empKMod_{(1)}\\
  0\\
  \vdots\\
  0
 \end{bmatrix}
 A^{t_j}\\
 = 
 \begin{bmatrix}
  k_{11}\la_1^{t_j} & \binom{t_j}{1}\la_1^{t_j-1} + k_{12}\la_1^{t_j}  & \cdots & k_{1\ncent}\la_{\JorMul}^{t_j}\\
  0 & 0 & \cdots & 0\\
  \vdots & \vdots & \ddots & 0\\
  0 & 0 & \cdots & 0
 \end{bmatrix}
\end{equation*}
Therefore, following some more elementary row operations encoded by $V\in\R^{\ncent\otime\times\ncent\otime}$, we get that
\begin{equation*}
 V \begin{bmatrix}
    \empKMod & \cdots & \empKMod
   \end{bmatrix}
   \begin{bmatrix}
    A^{t_1}\\
    \vdots\\
    A^{t_{\otime}}
   \end{bmatrix}
   = 
   \begin{bmatrix}
    \tilde{k}_{11}\la_1^{t_1} & \cdots & \tilde{k}_{1\ncent}\la_{\JorMul}^{t_1}\\
    \tilde{k}_{11}\la_1^{t_2} & \cdots & \tilde{k}_{1\ncent}\la_{\JorMul}^{t_2}\\
    \vdots & \ddots & 0\\
    \tilde{k}_{11}\la_1^{t_{\otime}} & \cdots 
    & \tilde{k}_{1\ncent}\la_{\JorMul}^{t_{\otime}}\\
    \mathbf{0} & \cdots & \mathbf{0}
   \end{bmatrix}\\
   = 
   \begin{bmatrix}
    \boldsymbol{\Phi}\\
    \widehat{\mathbf{0}}
   \end{bmatrix}.
\end{equation*}
If the individual entries $\tilde{k}_{1i}$ are nonzero, and the Jordan block diagonals have nonzero eigenvalues, the columns of $\boldsymbol\Phi$
become linearly independent. Therefore, if $\otime \geq \ncent$, the column rank of $\Obs_{\Tset}$ is $\ncent$, which results in an observable system.

To extend this proof to matrices $A = \JorP\JorLa\JorP^{-1}$, note that 
\begin{equation*}
 \Obs_{\Tset} = 
 \begin{bmatrix}
  \empK A^{t_1}\\
  \cdots\\
  \empK A^{t_\otime}
 \end{bmatrix}
 =
  \begin{bmatrix}
  \empK \JorP\JorLa^{t_1}\JorP^{-1}\\
  \cdots\\
  \empK \JorP\JorLa^{t_\otime}\JorP^{-1}.
 \end{bmatrix}
 =
 \begin{bmatrix}
  \empK & \cdots & \empK
 \end{bmatrix}
 \boldsymbol{\JorP}
 \boldsymbol{\JorLa}^t
 \boldsymbol{\JorP^{-1}},
\end{equation*} 
where $\boldsymbol{\JorP}\in\R^{\ncent\otime\times\ncent\otime}$, $\boldsymbol{\JorLa}^t\in\R^{\ncent\otime\times\ncent\otime}$, and
$\boldsymbol{\JorP^{-1}}\in\R^{\ncent\otime\times\ncent\otime}$ are the block diagonal matrices associated with the system. 
Since $\boldsymbol{\JorP}$ is an invertible matrix, the conclusions about the column rank drawn before still hold, and the system is observable. 
\end{proof}
When the eigenvalues of the system matrix are repeated, it is not enough for $\empK$ to be shaded. The next proposition proves a lower bound on the number of observations required. 
\begin{proposition}\label{prop:2}
 Suppose that the conditions in Proposition \ref{prop:1} hold, with the relaxation that
 the Jordan blocks $\begin{bmatrix}\JorLa_1 & \JorLa_2 &\cdots & \JorLa_{\JorMul}\end{bmatrix}$ may have 
 repeated eigenvalues. Let $\nevals$ be the number of unique eigenvalues of $A$, and let $\geomMult(\eval_i)$ denote the geometric multiplicity of eigenvalue $\eval_i$. 
 Then there exist kernels $\kernel(x,y)$ such that the lower bound $\minmeas$ on the number of sampling locations $\nsamp$ is given by the cyclic index of $A$, which can be computed as 
 \begin{align}\eqlabel{geom_mult}
  \minmeas = \max_{1\leq i\leq\nevals}\geomMult(\eval_i).
 \end{align}
\end{proposition}
\begin{proof}
 We first prove the lower bound. WLOG, let $\empKShadFull$ have $\minmeas-1$ fully shaded, linearly independent rows, and write it as 
 \begin{align*}
  \empKShadFull &= \begin{bmatrix}
                    k_{11} & k_{12} & \cdots & k_{1\ncent} \\
                    \vdots & \vdots & \cdots & \vdots \\
                    k_{(\minmeas-1)1} & k_{(\minmeas-1)2} & \cdots & k_{(\minmeas-1)\ncent} 
                   \end{bmatrix}.
 \end{align*}
 Since the cyclic index is $\minmeas$, this implies that at least one eigenvalue, say $\eval$, has $\minmeas$ Jordan 
 blocks. 
 Define indices $j_1, j_2, \dots, j_{\minmeas} \in \{1,2,\dots,\ncent\}$ as the columns corresponding to the leading entries of the $\minmeas$ Jordan blocks corresponding to $\eval$. WLOG, let $j_1 = 1$.
 Using ideas similar to the last proof, we can write the observability matrix as
 \begin{align*}
  \Obs_{\Tset}
   &:= 
   \begin{bmatrix}
    \kernel_{11}\eval^{t_1}  & \cdots & \kernel_{1j_{\minmeas}}\eval^{t_1} & \cdots\\
    \vdots & \ddots &\vdots & \ddots\\
    \kernel_{11}\eval^{t_{\otime}}  & \kernel_{1j_{\minmeas}}\eval^{t_{\otime}} & \cdots\\
    \vdots & \ddots & \vdots & \ddots\\
    \kernel_{(\minmeas-1)1}\eval^{t_1}  \cdots & \kernel_{(\minmeas-1)j_{\minmeas}}\eval^{t_1} & \cdots\\
    \vdots & \ddots &\vdots & \ddots\\
    \kernel_{(\minmeas-1)1}\eval^{t_{\otime}}  & \cdots & \kernel_{(\minmeas-1)j_{\minmeas}}\eval^{t_{\otime}} & \cdots
   \end{bmatrix}.
 \end{align*}
 Define $\evalvec:= \begin{bmatrix}\eval^{t_1} & \eval^{t_2} & \cdots \eval^{t_{\otime}}\end{bmatrix}^T$. 
 Then the above matrix becomes 
 \begin{align*}
  \Obs_{\Tset}
   &:= 
   \begin{bmatrix}
    \kernel_{11}\evalvec  & \cdots & \kernel_{1j_2}\evalvec & \cdots & \kernel_{1j_{\minmeas}}\evalvec & \cdots\\
    \vdots & \ddots & \vdots & \ddots &\vdots & \ddots\\
    \kernel_{(\minmeas-1)1}\evalvec  & \cdots & \kernel_{(\minmeas-1)j_2}\evalvec & \cdots & \kernel_{(\minmeas-1)j_{\minmeas}}\evalvec & \cdots
   \end{bmatrix}.
 \end{align*}
 We need to show that one of the columns above can be written in terms of the others. This is equivalent to solving the linear system
 \begin{align*}
  \begin{bmatrix}
   \kernel_{1j_1}\\
   \kernel_{2j_1}\\
   \vdots\\
   \kernel_{(\minmeas-1)j_1}
  \end{bmatrix}
  &=
  \begin{bmatrix}
   \kernel_{1j_2} & \cdots & \kernel_{1j_{\minmeas}}\\
   \kernel_{2j_2} & \cdots & \kernel_{2j_{\minmeas}}\\
   \vdots & \ddots & \vdots\\
   \kernel_{(\minmeas-1)j_2} & \cdots & \kernel_{(\minmeas-1)j_{\minmeas}}\\
  \end{bmatrix} 
  \begin{bmatrix}
   c_1\\
   c_2\\
   \vdots\\   
   c_{(\minmeas-1)}
  \end{bmatrix}. 
 \end{align*}
 Suppose the kernel matrix on the RHS is generated from the Gaussian kernel. From \cite{micchelli1984interpolation}, 
 it's known that every principal minor of a Gaussian kernel matrix is invertible, which implies that $\Obs_{\Tset}$ cannot be observable. 
\end{proof}
We now prove a sufficient condition for the observability of a system with repeated eigenvalues, but with the condition that the Jordan blocks are trivial. 
\begin{proposition}\label{prop:3}
 Suppose that the conditions in Proposition \ref{prop:1} hold, with the relaxation that
 the Jordan blocks $\begin{bmatrix}\JorLa_1 & \JorLa_2 &\cdots & \JorLa_{\JorMul}\end{bmatrix}$ may have 
 repeated eigenvalues, and where $\JorLa_i$ are single-dimensional. Let $\minmeas$ be the cyclic index of $A$.
 We define 
 \begin{align}\eqlabel{empKShadFull}
  \empKShadFull = \begin{bmatrix}
                    \empK^{(1)}\\
                    \vdots\\
                    \empK^{(\minmeas)}
                  \end{bmatrix}
 \end{align}
 as the \emph{$\minmeas$-shaded matrix} which consists of $\minmeas$ shaded matrices with the property that any subset of
 $\minmeas$
 columns in the matrix are linearly independent from each
 other. Then system \eqref{k_measure} is observable if $\Tset$ has distinct values, and $|\Tset| \geq \ncent$.
\end{proposition}
\begin{proof}
 A cyclic index of $\minmeas$ for this system implies that there exists an eigenvalue $\eval$ that's repeated $\minmeas$
 times. WLOG, let $\empKShadFull$ have $\minmeas$ fully shaded, linearly independent rows, and, assume that the column indices corresponding to this eigenvalue are $\{1,2,\dots,\minmeas\}$. Define 
 $\evalvec_i:= \begin{bmatrix}\eval_i^{t_1} & \eval_i^{t_2} & \cdots \eval_i^{t_{\otime}}\end{bmatrix}^T$. Then
 \begin{align*}
  \Obs_{\Tset}
   &:= 
   \begin{bmatrix}
    k_{11} \evalvec_1 & k_{12} \evalvec_2 & \cdots & k_{1\ncent} \evalvec_{\ncent}\\
    \vdots & \vdots & \ddots & \vdots\\
    k_{\minmeas 1} \evalvec_1 & k_{\minmeas 2} \evalvec_2 & \cdots & k_{\minmeas \ncent} \evalvec_{\ncent}.
   \end{bmatrix}
 \end{align*}
 Let $\evalvec_1 = \evalvec_2 = \cdots \evalvec_{\minmeas} := \evalvec$. 
 Focusing on these first $\minmeas$ columns of this matrix, this implies that
 we need to find constants $c_1,c_2,\dots, c_{\minmeas-1}$ s.t.
 \begin{align*}
  \begin{bmatrix}
   k_{11}\\
   \vdots\\
   k_{\minmeas 1}
  \end{bmatrix}
  &= 
  c_1 
  \begin{bmatrix}
   k_{12}\\
   \vdots\\
   k_{\minmeas 2}
  \end{bmatrix}
  + \cdots + 
  c_{\minmeas-1} 
  \begin{bmatrix}
   k_{1\minmeas}\\   
   \vdots\\
   k_{\minmeas \minmeas}. 
  \end{bmatrix}
 \end{align*}
 However, these columns are linearly independent by assumption, and thus no such constants exist, implying that $\Obs_{\Tset}$ is observable. 
\end{proof}
An example of a kernel such that any subset of $\minmeas$ columns in $\empKShadFull$ are linearly independent of each other is the Gaussian kernel evaluated on sampling locations $\sampSetLong$, where $x_i\in\dom\subset\R^d$, and $x_i\neq x_j$. 

We can reuse Propositions \ref{prop:1}, \ref{prop:2}, and \ref{prop:3} to prove kernel controllability results, because the structure of the control matrix $\empKCD$ in \eqref{k_measure_c} is also that of a kernel matrix.

\begin{figure}[t]
\centering
 \subfigure[Shaded kernel matrix (see Definition \ref{def:shaded})]{
 \includegraphics[width=0.4\columnwidth]{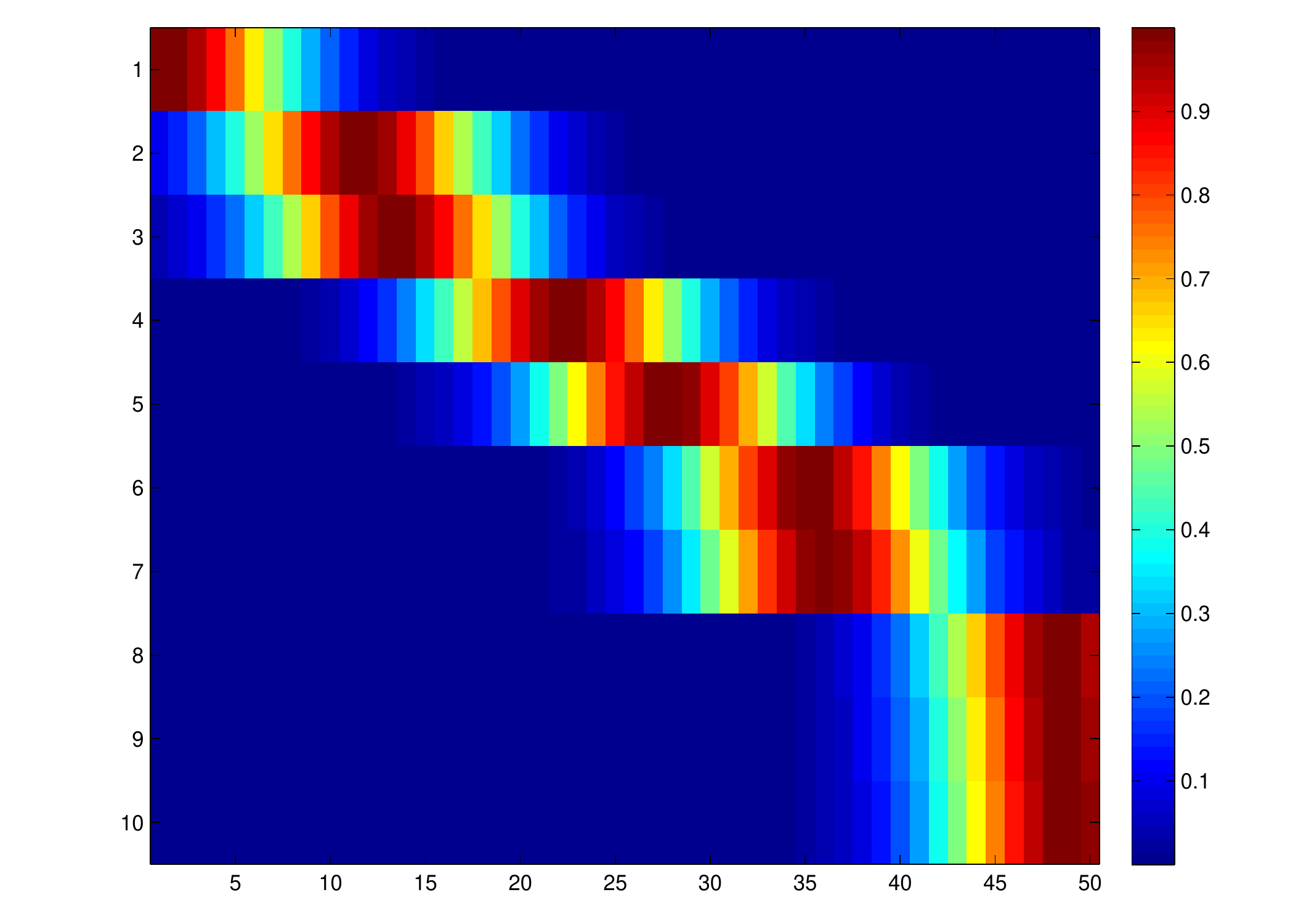} 
 \label{fig:shaded}}
 \quad
 \subfigure[2-shaded kernel matrix (see \eqref{empKShadFull})]{
 \includegraphics[width=0.4\columnwidth]{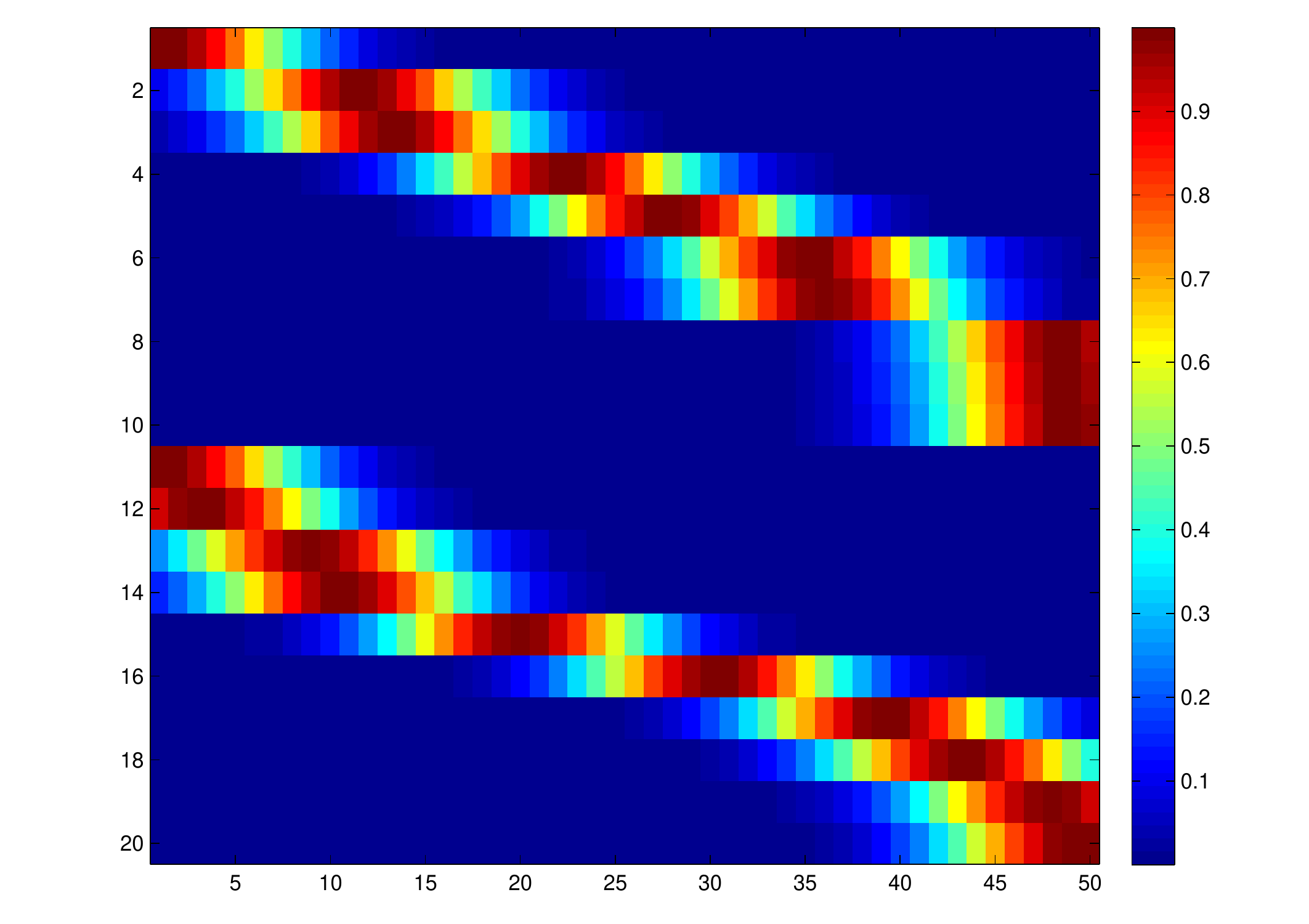} 
 \label{fig:shaded_multiple}}
 \caption{Pictorial representations of shaded kernel matrices.}
\label{shaded_matrix}
\vspace{-0.2in}
\end{figure}
\begin{figure}[t!]
\centering
\begin{algorithm}[H]
   \caption{Kernel Observer (Transition Learning)}
   \label{alg:egp_trans}
\begin{algorithmic}
\begin{footnotesize}
   \STATE {\bfseries Input:} Kernel $\kernel$, basis points $\shCent$, final time 
    step $T_f$. 
   \WHILE{$k \leq T_f$}
   \STATE $1)$ Sample data $\{y^i_k\}_{i=1}^{\nsamp}$ from $f(x,k)$. 
   \STATE $2)$ Estimate $\estweight_k$ via standard kernel inference procedure. 
   \STATE $3)$ Store weights $\estweight_k$ in matrix $\W\in\R^{\fdim\times T_f}$.
   \ENDWHILE
   \STATE Infer $\estsysop$ using method of choice (e.g. matrix least squares).
          Compute the covariance matrix  $\estcontrolop$ of the observed 
          weights $\W$. 
  \STATE {\bfseries Output:} estimated transition matrix $\estsysop$, predictive covariance    
         matrix $\estcontrolop$. 
\end{footnotesize}
\end{algorithmic}
\end{algorithm}
\vspace{-0.2in}
\begin{algorithm}[H]
   \caption{Kernel Observer (Estimation and Prediction)}
   \label{alg:egp_inf}
\begin{algorithmic}
\begin{footnotesize}
   \STATE {\bfseries Input:} Kernel $\kernel$, basis points $\shCent$, 
   estimated system matrix $\estsysop$, estimated covariance matrix $\estcontrolop$.
   \STATE {\bfseries Compute Observation Matrix:} Compute the cyclic index $\minmeas$ of $\estsysop$, and compute \eqref{empKShadFull}, by possibly iterating over $\sampSet = \sampSetLong$. 
   \STATE {\bfseries Initialize Observer:} Use $\estsysop$, $\estcontrolop$, and $\empKShadFull$ to initialize a state-observer (e.g. Kalman filter (KF)) on $\fspaceC$.
   \WHILE{ measurements available }   
     \STATE 1) Sample data $\{y^i_k\}_{i=1}^{\nsamp}$ from $f(x,k)$.     
     \STATE 2) Propagate KF estimate $\estweight_{k+1}$ 
               forward to time $t_f$, correct using measurement feedback with $\{y^i_k\}_{i=1}^{\nsamp}$. 
     \STATE 3) Output predicted function $\widehat{f}(x,k+1)$ and predictive covariance of KF.
   \ENDWHILE   
\end{footnotesize}
\end{algorithmic}
\end{algorithm}
\vspace{-0.2in}
\end{figure}
\begin{figure}[t]
\begin{algorithm}[H]
   \caption{Kernel Controller}
   \label{alg:egp_control}
\begin{algorithmic}
\begin{footnotesize}
   \STATE {\bfseries Input:} Kernel $\kernel$, basis points $\shCent$,
   estimated system matrix $\estsysop$, estimated covariance matrix $\estcontrolop$, and function $f_{\text{ref}}$ 
   to drive initial function to.    
   \STATE {\bfseries Initialize Observer:} (see Algorithm \ref{alg:egp_inf}). 
   \STATE {\bfseries Initialize Controller:} Use Jordan decomposition of $\estsysop$ to obtain  control locations $\controlCent$, compute kernel matrix $\empKCD\in\R^{\ncontrol\times\ncent}$ between $\controlCent$ and $\shCent$, and initialize controller (e.g. LQR) utilizing $(\estsysop, \estcontrolop)$.
   \WHILE{ measurements available }   
     \STATE 1) Sample data $\{y^i_k\}_{i=1}^{\nsamp}$ from $f(x,k)$.    
     \STATE 2) Utilize observer to estimate $\estweight_{k+1}$.
     \STATE 3) Use $\estweight_{k+1}$ and $f_{\text{ref}}$ as input to controller to get feedback. 
   \ENDWHILE   
\end{footnotesize}
\end{algorithmic}
\end{algorithm}
\vspace{-0.2in}
\end{figure}

\section{Experimental Results}\label{sec:experimental_results}
We report experimental results on controlling synthetic and modeling real-world data. 
All experiments were performed using MATLAB on a laptop running Ubuntu 14.04 with $8$ GB of RAM, and an Intel core i7 processor. 

\subsection{Prediction of global ocean surface temperature} 
We first analyzed the feasibility of this modeling approach on a large dataset: the $4$ km AVHRR Pathfinder project, which is a satellite monitoring global ocean surface temperature. This data was obtained from the National Oceanographic Data Center. 
The data consists of longitude-latitude measurements on a 2D domain $\dom\subset[-180,180]\times[-90,90]$;
this dataset is challenging, with measurements at over $37$ million coordinates, and several missing pieces of data. The goal was to learn the day and night temperature models $f_k(x,y)\in\fspaceC$, where $\fspaceC$ was generated using the Gaussian kernel $\kernel(x,y) = e^{-(\|x-y\|^2/2\s^2)}$. 
We first did a search for the ideal bandwidth $\s$ for a $304$-dimensional sparse Gaussian process model with a Gaussian kernel. The set of atoms $\Atoms$ was determined through a linear independence test based sparsification algorithm \cite{csato2002sparse}. 
Once the parameters were chosen, a budgeted GP was learned for each date, resulting in weight vectors $\weight_i, \ i\in\{1,2,\dots,365\}$.  
We used Algorithm \ref{alg:egp_trans} to infer $\estsysop$, and applied Algorithm \ref{alg:egp_inf} with $\nsamp \in \{280,500,1000,2000\}$ chosen randomly in the $\dom$ to track the system state given a random initial condition $\weight_0$. 
Figures \ref{fig:pathfinder_errors_boxplots_day} and \ref{fig:pathfinder_errors_boxplots} show a comparison of the deviation in percentage of the estimated values from the real data, averaged over all the days. 
As can be seen, the observer enables the prediction of functional evolution \textit{without needing all the measurements (37 million)}, and performance comparable to sampling over all locations is obtained with sampling only over $2,000$ locations. Note that here, even though the system model is observable at $\nsamp=280$, since the dynamics are not truly linear in $\fspaceC$, we get better performance with more sampling locations.  
Finally, \ref{fig:pathfinder_tr_times_boxplots_day} and \ref{fig:pathfinder_tr_times_boxplots} show that the time required to estimate the state  
during function tracking with kernel observer are an order of magnitude better than retraining the model every time step (``original'' in the figure), with comparable performance.
\begin{figure}[tbh] 
    \centering
    \subfigure[Pathfinder raw data on a fixed daty]{
    \includegraphics[width=0.40\columnwidth]{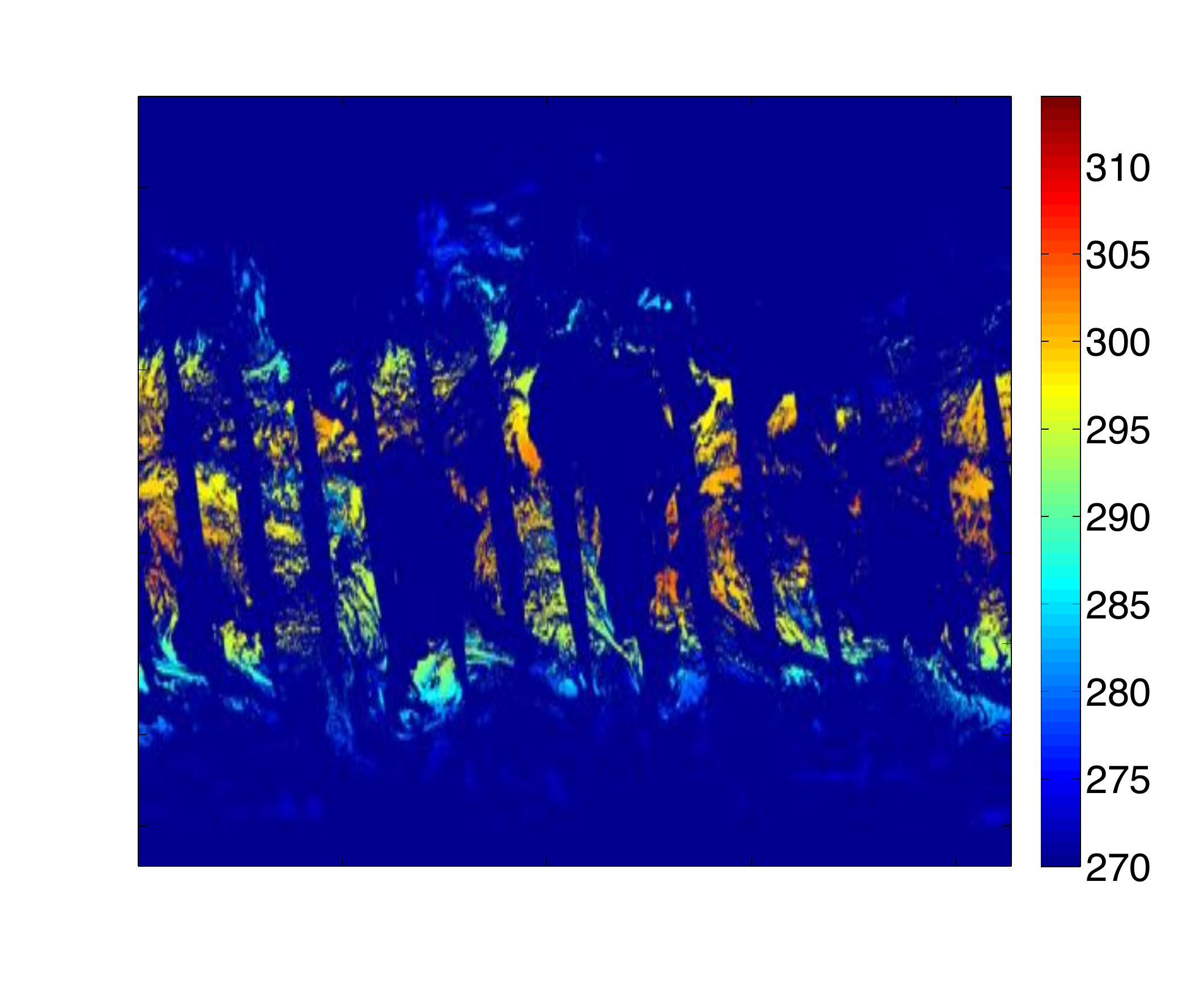} 
    \label{fig:pathfinder_raw_data} }
    \quad
    \subfigure[Pathfinder kernel observer estimate]
    {
    \includegraphics[width=0.40\columnwidth]{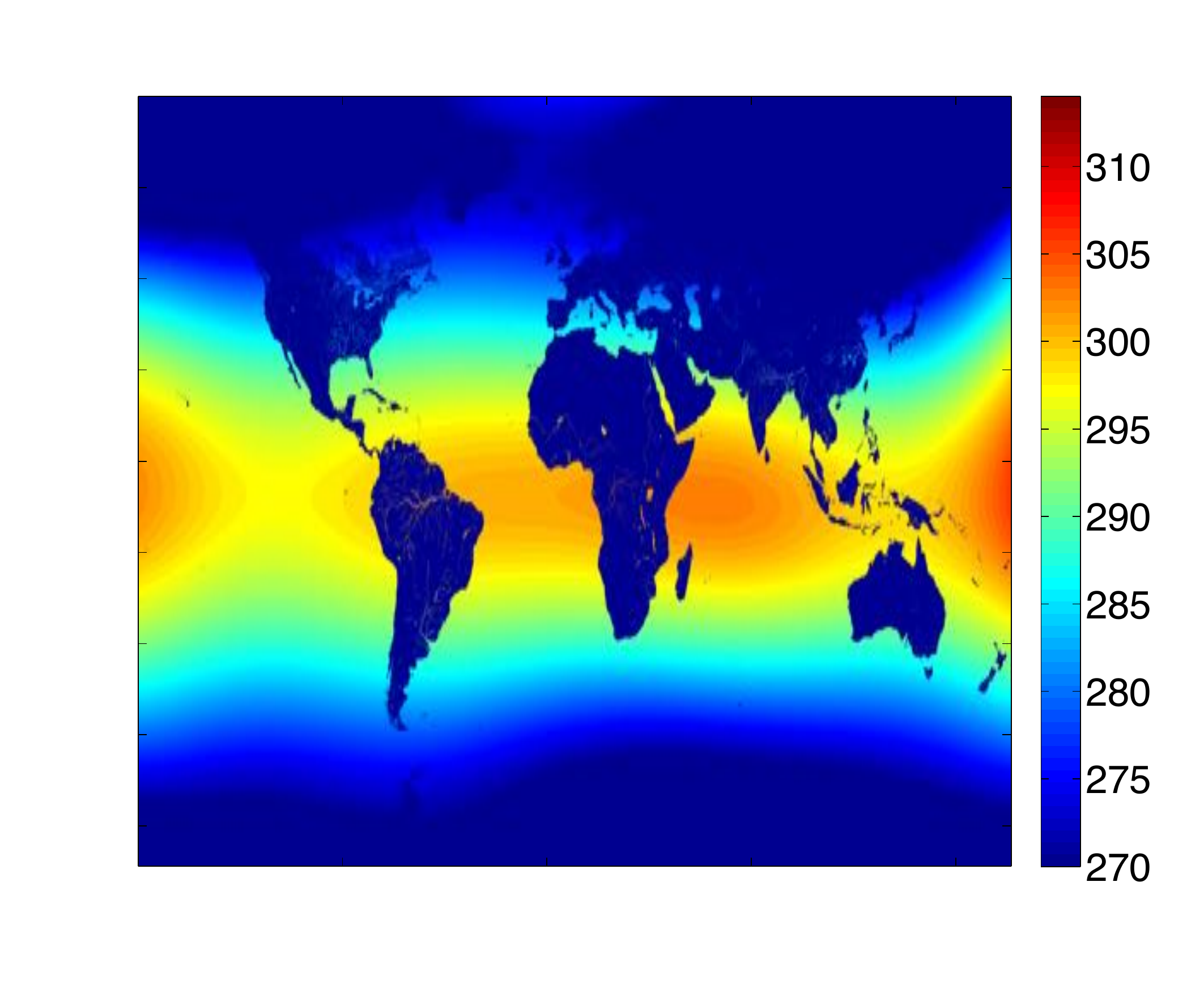} 
    \label{fig:pathfinder_ko_est}}
    \caption{Pathfinder raw data and kernel observer estimate, computed on data from $05/01/2012$.}     
        \vspace{-0.2in}
\end{figure} 
\begin{figure}[tbh] 
    \centering
    \subfigure[Estimation error (day)]{
    \includegraphics[width=0.40\columnwidth]{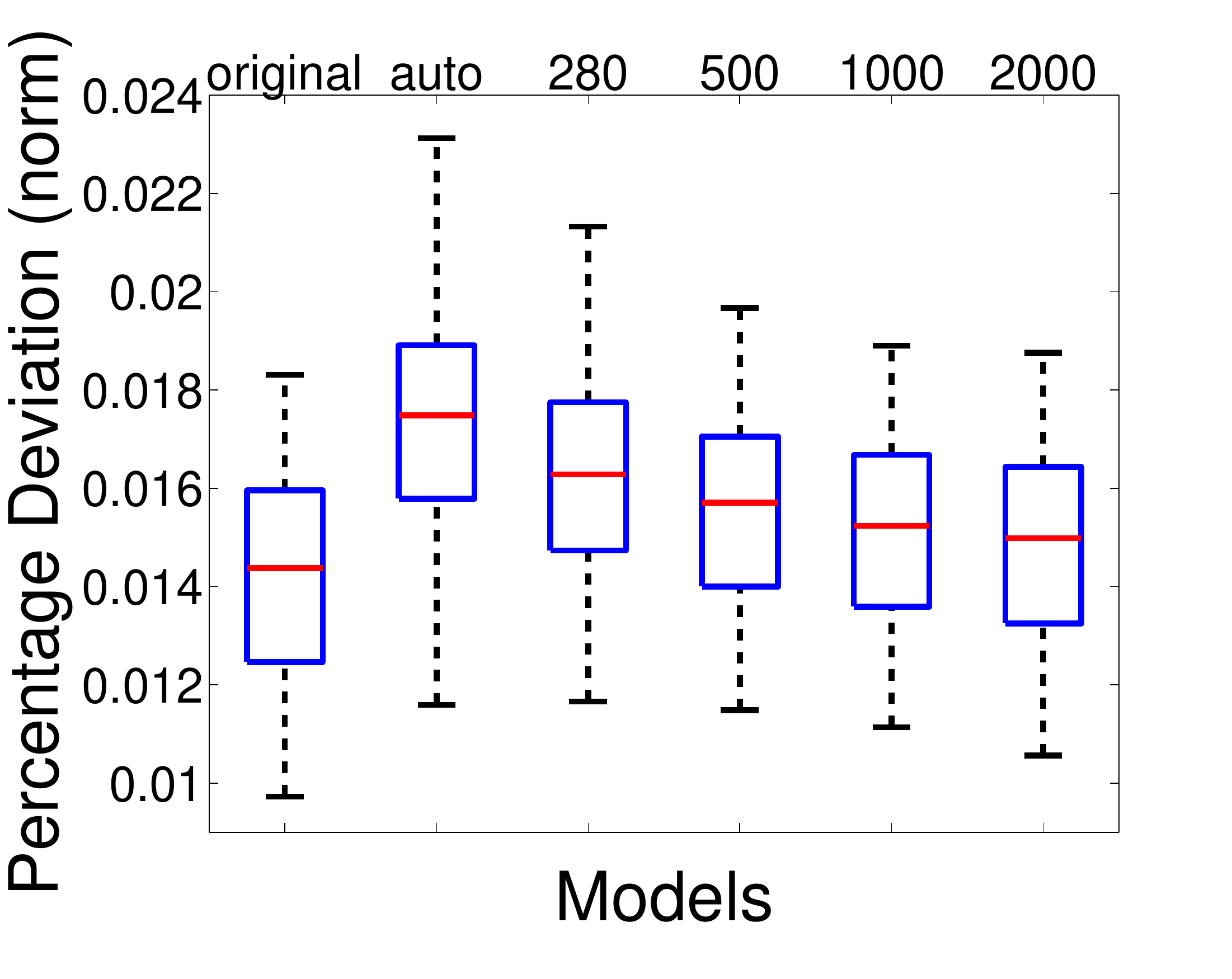} \label{fig:pathfinder_errors_boxplots_day} } 
    \subfigure[Estimation time (day)]
    {
    \includegraphics[width=0.40\columnwidth]{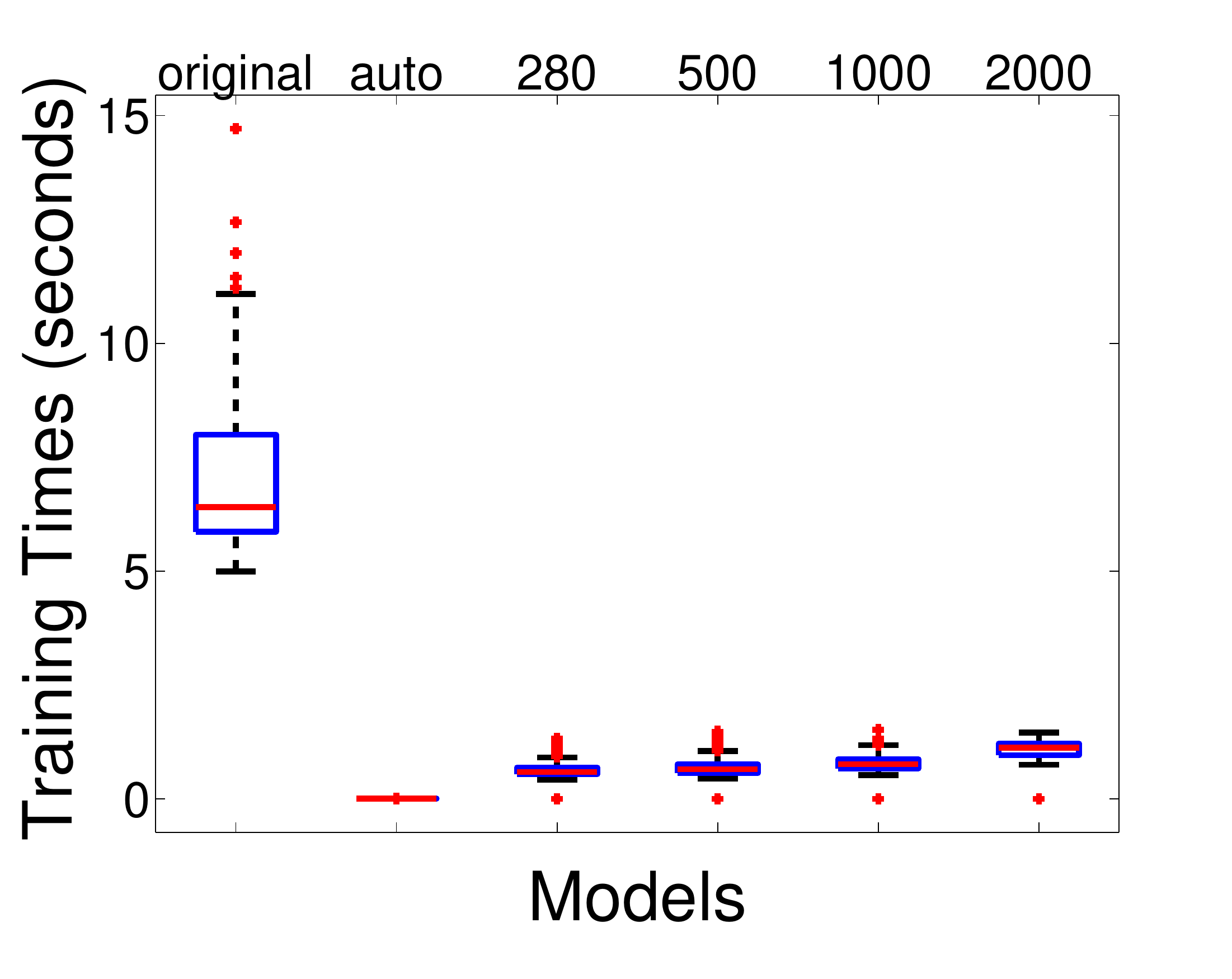} \label{fig:pathfinder_tr_times_boxplots_day}}
    \\
    \subfigure[Estimation error (night)]{
    \includegraphics[width=0.40\columnwidth]{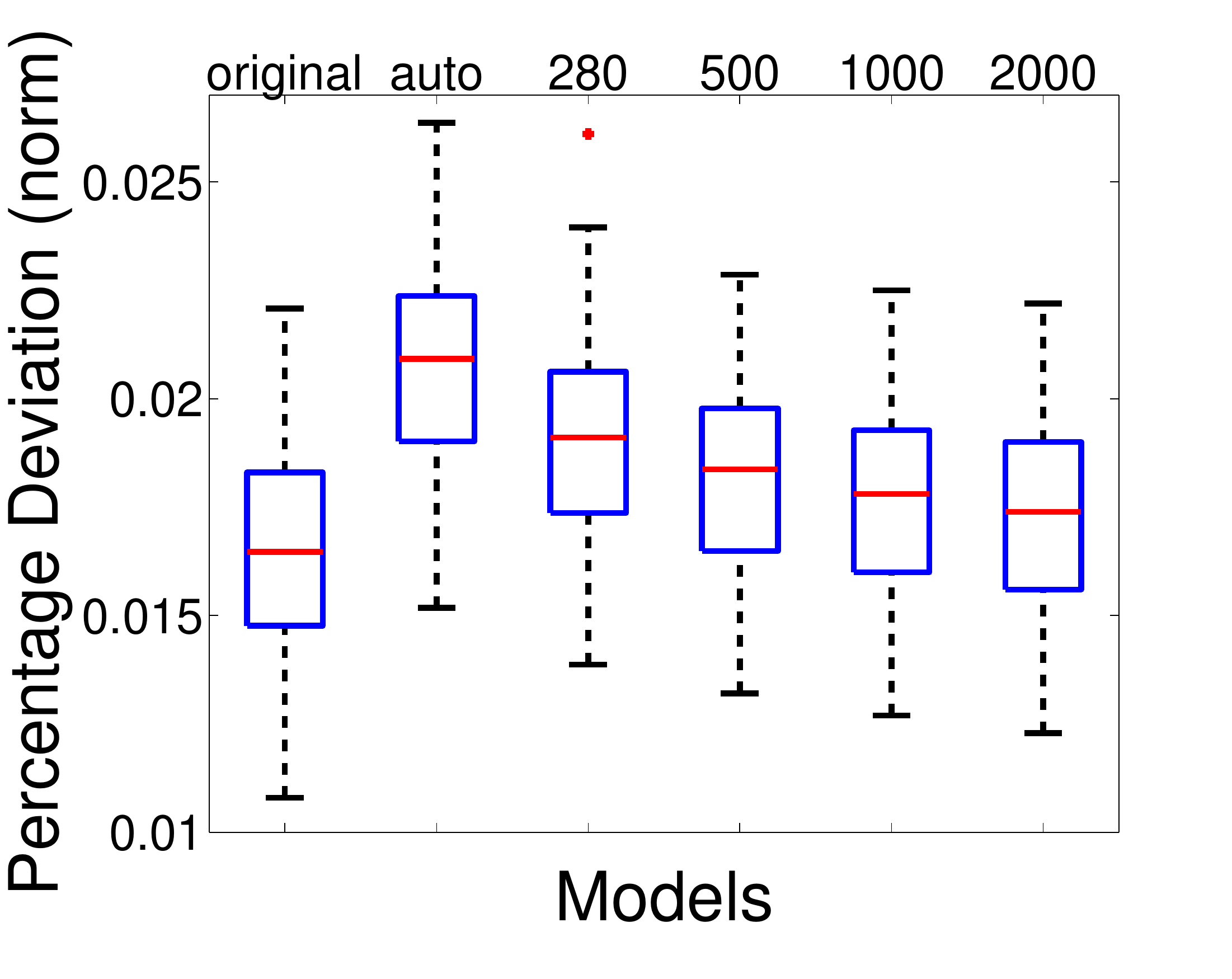} \label{fig:pathfinder_errors_boxplots} } 
    \subfigure[Estimation time (night)]
    {
    \includegraphics[width=0.40\columnwidth]{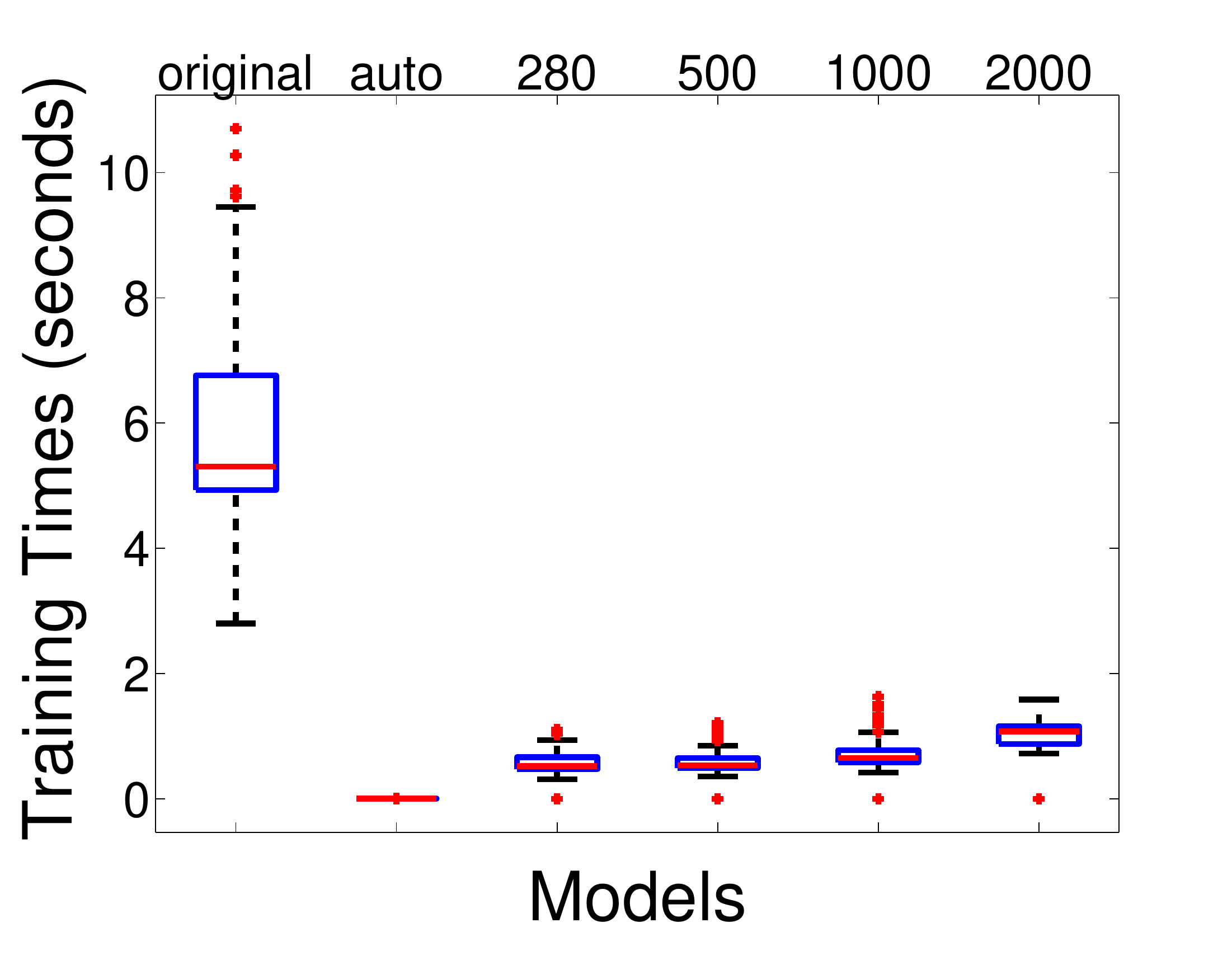} \label{fig:pathfinder_tr_times_boxplots}}
    \caption{Performance of kernel observer over Pathfinder satellite 2012 data with different numbers of observations.}     
        \vspace{-0.2in}
\end{figure}

\subsection{Control of a linear PDE} 
We then employed kernel controllers for controlling an approximation to the scalar diffusion equation $u_t = bu_{xx}$ on the domain $\dom=[0,1]$, with $b=0.25$. The solution to this equation is infinite-dimensional, so we chose a kernel $\kernel(x,y) = e^{-(\|x-y\|^2/2\s^2)}$, and a set of atoms  $\Atoms=\shCentLong$, $c_i\in\dom$, with $\ncent = 25$ generating $\fspaceC$, the space approximating $\fspace$, and another set of atoms $\AtomsControl=\{\fmap(d_1),\dots,\fmap(d_{\ncontrol})\}$, $d_j\in\dom$,
$\ncontrol=13$, generating the control space $\fspaceD$. The number of, and the location of the observations was chosen to be the same as that of the actuation locations $d_j$. First, tests (not reported here) were conducted to ensure that the solution to the diffusion equation is well approximated in $\fspaceC$. Algorithm \ref{alg:egp_trans} was then used to infer $\estsysop$. Figure \ref{fig:uncontrolled_pde} shows an example of an initial function $f_{\text{init}}$ evolving according to the PDE.  A reference function $f_{\text{ref}}\in\fspaceC$ was chosen to drive $f_{\text{init}}$ to $f_{\text{ref}}$ under the action of the PDE. Finally, Algorithm \ref{alg:egp_control} was used to control the PDE. Figure \ref{fig:controlled_pde} shows $f_{\text{init}}$ being driven to $f_{\text{ref}}$, while Figure \ref{fig:controlled_pde_error} shows the absolute value of the error between $f_k $ and $f_{\text{ref}}$ as a function of time. 
\begin{figure}[tbh] 
    \centering
    \subfigure[Evolution of initial function $f_{\text{init}}$ according to diffusion equation.]{
    \includegraphics[width=0.31\columnwidth]{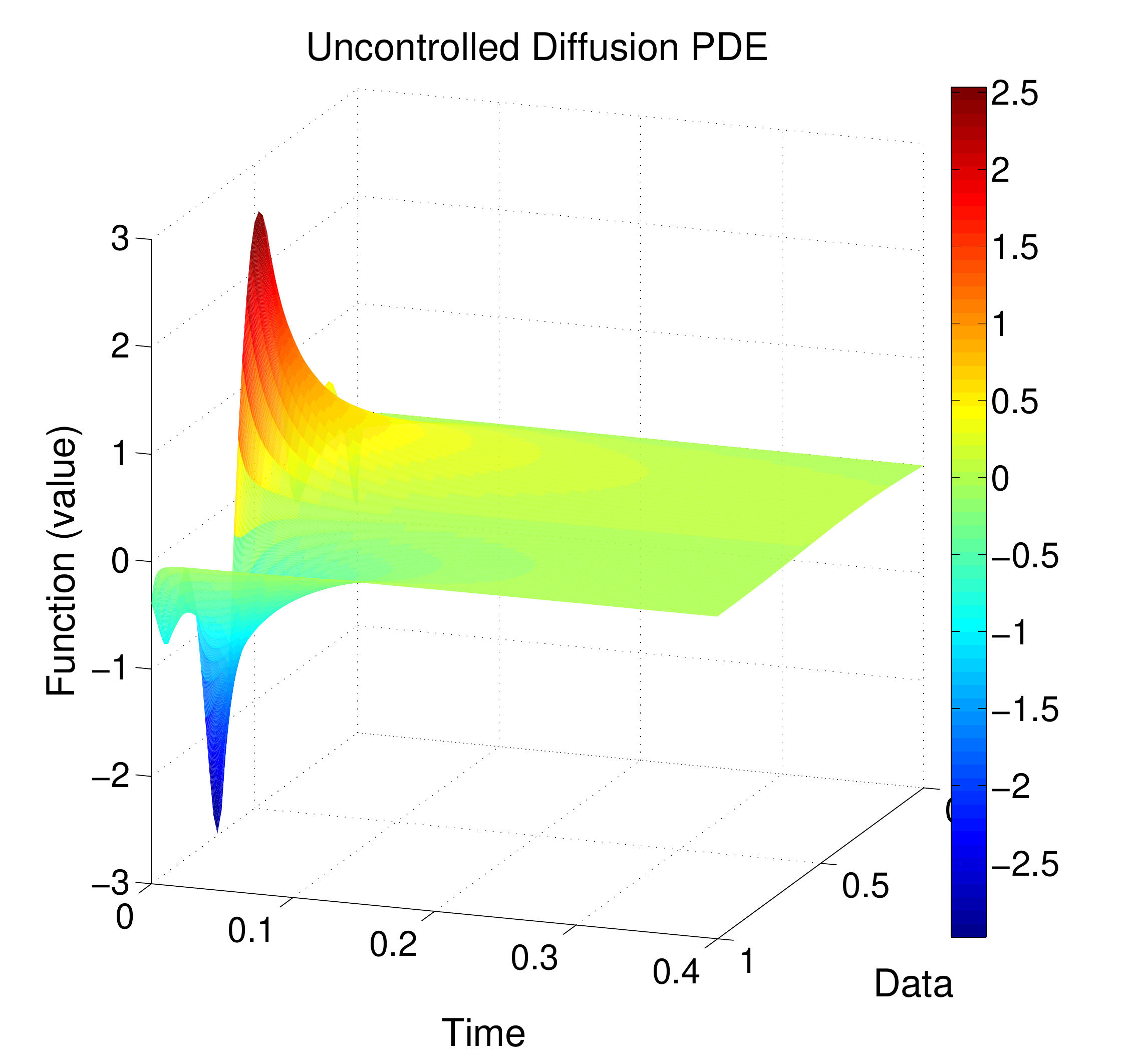} 
    \label{fig:uncontrolled_pde}}
    \subfigure[Initial function $f_{\text{init}}$ driven to $f_{\text{ref}}$ using kernel controller.]
    {
    \includegraphics[width=0.31\columnwidth]{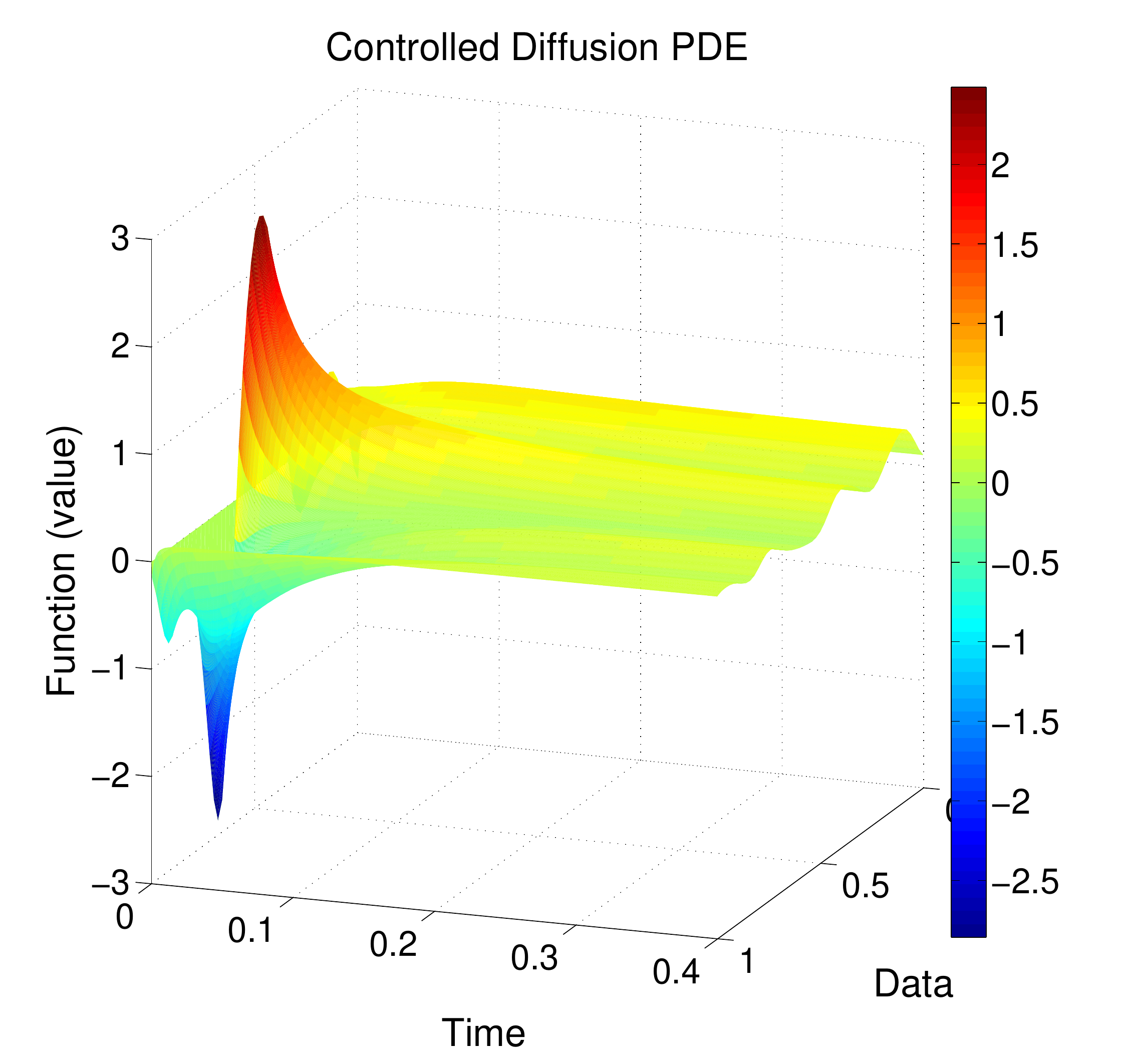} 
    \label{fig:controlled_pde}}    
    \subfigure[Error in absolute value between controlled pde and $f_{\text{ref}}$. ]{
    \includegraphics[width=0.31\columnwidth]{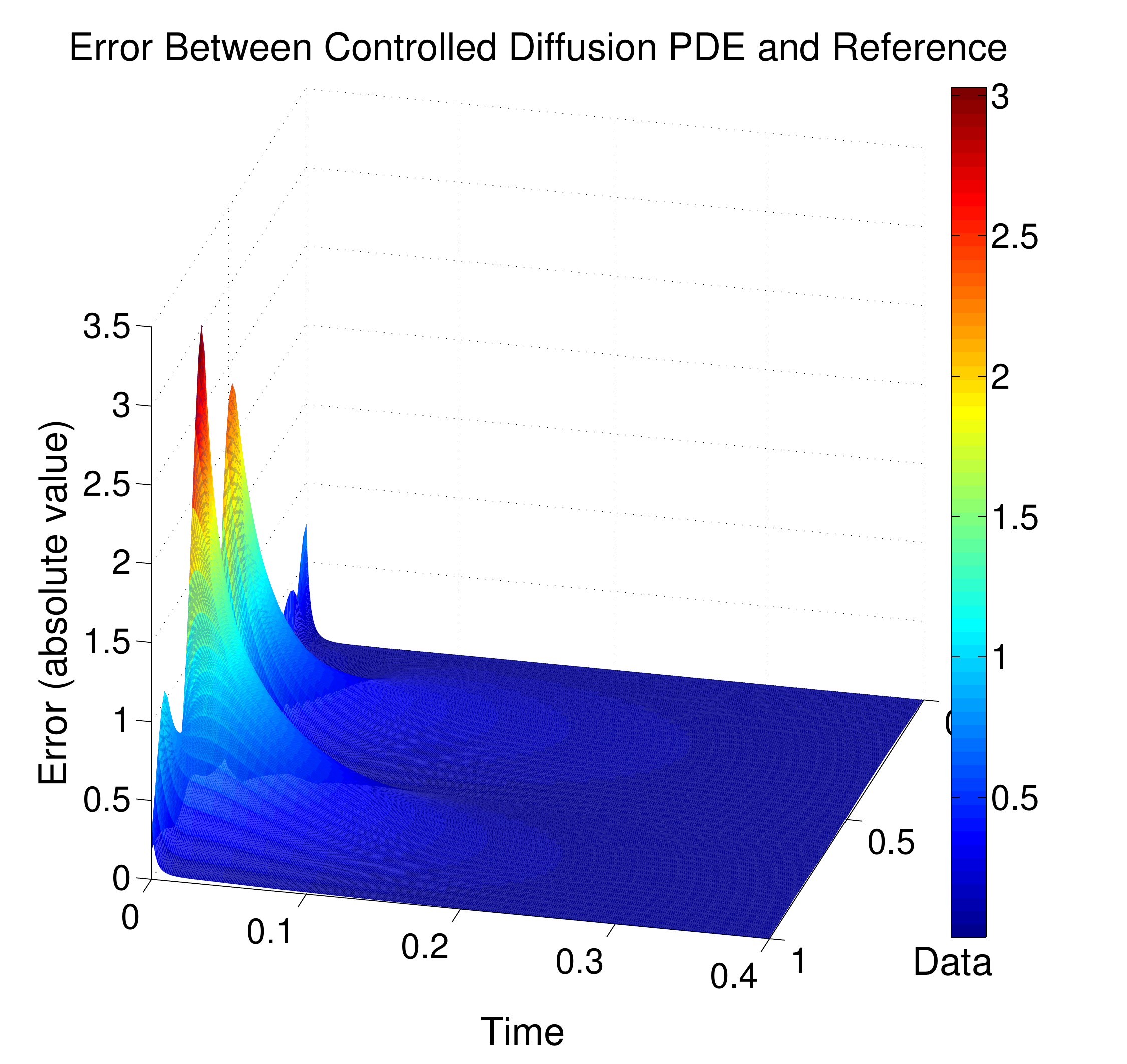} 
    \label{fig:controlled_pde_error} }
    \caption{Demonstration of the control of a linear diffusion equation.}     
    \vspace{-0.3in}
\end{figure} 

\section{Conclusions}
In this paper we presented a systems theoretic approach to the problem of modeling, estimating, and controlling complex spatiotemporally evolving phenomena. Our approach focused on developing a predictive model of spatiotemporal evolution by layering a dynamical systems prior over temporal evolution of weights of a kernel model. The resulting model can approximate PDE evolution, while it has the form of a finite state linear dynamical system. 
The lower bounds on the number of sampling and actuation locations provided in this paper are non-conservative, as such they provide direct guidance in ensuring robust real-world sensor network and actuation matrix design that must also account for fault-tolerance and reliability considerations.

\end{document}